\documentclass[conference]{IEEEtran}

\usepackage{booktabs}   \usepackage{subcaption} \usepackage{proof}
\usepackage{xcolor}

\usepackage{colortbl}

\usepackage{amsmath}
\usepackage{stmaryrd}
\usepackage{amsfonts}
\usepackage{amssymb}
\usepackage{cite}
\usepackage[all]{xy} \SelectTips{cm}{}  \newdir{ >}{{}*!/-8pt/@{>}}  
\newdir{|>}{!/1.6pt/@{|}*:(1,-.2)@^{>}*:(1,+.2)@_{>}}

\usepackage{todonotes}

\usepackage{mathtools}
\usepackage{hyperref}
\usepackage[capitalize, nameinlink]{cleveref}

\usepackage{aliascnt}
\usepackage{apxproof}
\renewcommand{\appendixprelim}{}
\renewcommand{\appendixsectionformat}[2]{Omitted Proofs for Section #1}
\newcommand{\myqed}{\qed}
\newcommand{\mymyqed}{\tag*{$\qedsymbol$}}
\theoremstyle{plain}
\newtheoremrep{theorem}{Theorem}[section]
\Crefname{theorem}{Thm.}{Thm.}
\newaliascnt{propositioncnt}{theorem}
\newtheoremrep{proposition}[propositioncnt]{Proposition}
\Crefname{propositioncnt}{Prop.}{Prop.}
\newaliascnt{lemmacnt}{theorem}
\newtheoremrep{lemma}[lemmacnt]{Lemma}
\Crefname{lemmacnt}{Lem.}{Lem.}
\newaliascnt{corollarycnt}{theorem}
\newtheoremrep{corollary}[corollarycnt]{Corollary}
\Crefname{corollarycnt}{Cor.}{Cor.}
\newaliascnt{factcnt}{theorem}
\newtheorem{fact}[factcnt]{Fact}
\Crefname{factcnt}{Fact}{Fact}
\newaliascnt{assumptioncnt}{theorem}
\newtheorem{assumption}[assumptioncnt]{Assumption}
\Crefname{assumptioncnt}{Asm.}{Asm.}
\theoremstyle{definition}
\newaliascnt{remarkcnt}{theorem}
\newtheorem{remark}[remarkcnt]{Remark}
\Crefname{remarkcnt}{Rem.}{Rem.}
\newaliascnt{notationcnt}{theorem}
\newtheorem{notation}[notationcnt]{Notation}
\Crefname{notationcnt}{Notation}{Notation}
\newtheorem*{claim}{Claim}

\theoremstyle{definition}
\newaliascnt{definitioncnt}{theorem}
\newtheoremrep{definition}[definitioncnt]{Definition}
\Crefname{definitioncnt}{Def.}{Def.}
\newaliascnt{examplecnt}{theorem}
\newtheoremrep{example}[examplecnt]{Example}
\Crefname{examplecnt}{Ex.}{Ex.}

\crefformat{section}{{\S}#2#1#3}
\Crefname{table}{Table}{Table}
\Crefname{equation}{}{}

\usepackage{mathrsfs}

\usepackage{wrapfig}
\usepackage{comment}
\specialcomment{auxproof}
{\mbox{}\newline\textbf{BEGIN: AUX-PROOF}\dotfill\newline}
{\mbox{}\newline\textbf{END: AUX-PROOF}\dotfill\newline}
\excludecomment{auxproof}

\newcommand{\op}{\mathrm{op}}
\newcommand{\Op}[1]{{#1}^{\op}}
\newcommand{\Cat}[1]{\mathbb{#1}}
\newcommand{\arrow}{\mathbin{\to}}
\newcommand{\ol}[1]{\overline{#1}}

\newcommand{\co}{\mathop{\circ}}
\newcommand{\Id}{\mathrm{Id}}
\newcommand{\tuple}[1]{\langle{#1}\rangle}
\newcommand{\CC}{\Cat{C}}
\newcommand{\EE}{\Cat{E}}
\newcommand{\DD}{\Cat{D}}
\newcommand{\id}{\mathrm{id}}
\newcommand{\place}{\underline{\phantom{n}}\,} \newcommand{\bec}[3]{#1 ~ & #2 &\text{#3}\\}
\newcommand{\becncr}[3]{#1 ~ & #2 &\text{#3}}

\newcommand{\adjunction}[3]{
	\ar@<.4pc>[#1]^-{#2}
	\ar@{}[#1]|-*=0[@]{\bot}
	\ar@<-.4pc>@{<-}[#1]_-{#3}
}
\newcommand{\adjointEquivalence}[3]{
	\ar@<.4pc>[#1]^-{#2}
	\ar@{}[#1]|-*=0[@]{\simeq}
	\ar@<-.4pc>@{<-}[#1]_-{#3}
}

\newcommand{\CLatw}{\mathbf{CLat}_\sqcap}
\newcommand{\Set}{\mathbf{Set}}
\newcommand{\Cats}{\mathbf{Cat}}
\newcommand{\PMetb}{\mathbf{PMet}_{1}}
\newcommand{\EqRel}{\mathbf{EqRel}}
\newcommand{\ERel}{\mathbf{ERel}}
\newcommand{\Meas}{\mathbf{Meas}}
\newcommand{\EqRelMeas}{\mathbf{EqRel}_{\Meas}}
\newcommand{\Pred}{\mathbf{Pred}}
\newcommand{\Pre}{\mathbf{Pre}}
\newcommand{\SDist}{\mathcal{D}_{\le 1}}
\newcommand{\Pow}{\mathcal{P}}
\newcommand{\CA}[1]{\mathbf{CoAlg}(#1)}
\newcommand{\cart}[1]{\widetilde{#1}}

\newcommand{\esit}{\mathscr{S}}
\newcommand{\bOmega}{\underline{\Omega}}
\newcommand{\rank}{\mathrm{rank}}
\newcommand{\codlif}[3]{ {\ol{#1}}^{ {#2},{#3} } }
\newcommand{\LEq}[2]{\mathsf{LE}_{#1}(#2)}
\newcommand{\CBisim}[3]{\mathsf{Bisim}^{ {#1},{#2} }(#3)}
\newcommand{\sem}[1]{\llbracket #1 \rrbracket}

\newcommand{\depth}{\mathrm{depth}}

\newcommand{\SGiry}{\mathcal{G}_{\le 1}}
\newcommand{\Eq}{\mathrm{Eq}}

\newcommand{\Unif}{\mathbf{Unif}}
\newcommand{\RR}{\mathbb{R}}

\newcommand{\theory}{\mathop{\mathrm{t\kern-0.07em h}}\nolimits}

\begin{document}
\IEEEoverridecommandlockouts
\title{Expressivity of Quantitative Modal Logics \\
	{\huge Categorical Foundations via Codensity and Approximation}\thanks{ The authors are
		supported by ERATO HASUO Metamathematics for Systems Design
		Project
(No.~JPMJER1603), JST.
	Thanks to Bart Jacobs for giving helpful comments on a draft of this paper.}}

\author{\IEEEauthorblockN{Yuichi Komorida\IEEEauthorrefmark{1}\IEEEauthorrefmark{2},
Shin-ya Katsumata\IEEEauthorrefmark{2},
Clemens Kupke\IEEEauthorrefmark{3}, 
Jurriaan Rot\IEEEauthorrefmark{4} and
Ichiro Hasuo\IEEEauthorrefmark{1}\IEEEauthorrefmark{2}}
\IEEEauthorblockA{\IEEEauthorrefmark{1}The Graduate University for Advanced Studies, SOKENDAI, Hayama, Japan}
\IEEEauthorblockA{\IEEEauthorrefmark{2}National Institute of Informatics, Tokyo, Japan}
\IEEEauthorblockA{\IEEEauthorrefmark{3}University of Strathclyde, United Kingdom}
\IEEEauthorblockA{\IEEEauthorrefmark{4}Radboud University, Nijmegen, The Netherlands}}

\IEEEoverridecommandlockouts
\IEEEpubid{\makebox[\columnwidth]{978-1-6654-4895-6/21/\$31.00~
		\copyright2021 IEEE \hfill} \hspace{\columnsep}\makebox[\columnwidth]{ }}

\maketitle

\begin{abstract}
A modal logic that is strong enough to fully characterize the behavior of a system is called expressive. Recently, with the growing diversity of systems to be reasoned about (probabilistic, cyber-physical, etc.), the focus shifted to quantitative settings which resulted in a number of expressivity results for quantitative logics and behavioral metrics. Each of these quantitative expressivity results uses a tailor-made argument; distilling the essence of these arguments is non-trivial, yet important to support the design of expressive modal logics for new quantitative settings. In this paper, we present the first categorical framework for deriving quantitative expressivity results, based on the new notion of approximating family. A key ingredient is the codensity lifting---a uniform observation-centric construction of various bisimilarity-like notions such as bisimulation metrics. We show that several recent quantitative expressivity results (e.g.\ by K\"{o}nig et al.\ and by Fijalkow et al.) are accommodated in our framework; a new expressivity result is derived, too, for what we call bisimulation uniformity.

\end{abstract}

\section{Introduction}
\label{sec:introductionX}

\paragraph{(Quantitative) Modal Logics and Their Coalgebraic Unification}
The role of different kinds of \emph{modal logics} is pervasive in computer science. Their principal functionality is to specify and reason about behaviors of state-transition systems. With the growing diversity of target systems (probabilistic, cyber-physical, etc.), the use of \emph{quantitative} modal logics---where truth values and logical connectives can involve real numbers---is increasingly common. For such logics, however, providing the necessary theoretical foundations takes a significant
effort and is often done individually for each variant. 

It is therefore desirable to establish unifying and abstract foundations once and for all, which readily instantiate to individual modal logics. This is the goal pursued by the study of \emph{coalgebraic modal logic}~\cite{Moss99,schr08:expr,Pattinson04,BonsangueK05,BonsangueK06,PavlovicMW06,Klin07}, which builds on the general categorical modeling of state-transition systems as \emph{coalgebras}~\cite{jacobs-coalg,Rutten00}.

\paragraph{Expressivity of Modal Logics}
When using a concrete modal logic, there are several important properties that we expect its metatheory to address, such as \emph{soundness} and \emph{completeness} of its proof system. In this paper, we are interested in the \emph{adequacy} and \emph{expressivity} properties of the logic. These properties are about comparison between 1) the expressive power of the logic, and 2) some notion of \emph{indistinguishability} that is inherent in the target state-transition systems.

A prototypical example of such notions of indistinguishability is \emph{bisimilarity}~\cite{Milner89}. Expressivity with respect to bisimilarity---that modal logic formulas can distinguish non-bisimilar states---is the classic result by Hennessy and Milner~\cite{HennessyM85}. Adequacy, the opposite of expressivity, means that semantics of modal formulas is invariant under bisimilarity, and holds in most modal logics. In contrast, expressivity is a desired property but not always true. Expressivity, when it holds, relies on a delicate balance between the choice of modal operators, the underlying propositional connectives, and the ``size'' of (branching of) the target state-transition systems. 

\paragraph{Quantitative Expressivity}
The aforementioned  interests in quantitative modal logics have sparked research efforts for \emph{quantitative expressivity}. In quantitative settings, the inherent indistinguishability notion in target systems is quantitative, too, typically formulated in terms of a \emph{bisimulation pseudometric} (``how much apart the two states are'') that refines the quantitative notion of bisimilarity (``if the two states are indistinguishable or not'')~\cite{Giaca90,DesharnaisGJP04}.  In the expressivity problem, such an indistinguishability notion is compared against the quantitative truth values of logical formulas.

Recent works that study quantitative expressivity include~\cite{DesharnaisGJP04,BreugelW05,ClercFKP19,WildSP018,KonigM18,WildS20}; they often involve coalgebraic generalization, too, since quantitative modal logics often have immediate variations.  Their quantitative expressivity proofs are much more mathematically involved compared to qualitative expressivity proofs. This is because the aforementioned balance between syntax and semantic equivalences is much more delicate.  Specifically, target systems are quantitative and thus exhibit \emph{continuity}  of behaviors, while logical syntax is inherently \emph{disconnected}, in the sense that each logical formula is an inductively defined and thus finitary entity. Expressivity needs to bridge these two seemingly incompatible worlds.

In order to do so, each of the expressivity proofs in~\cite{BreugelW05,ClercFKP19,WildSP018,KonigM18,WildS20} uses some kind of ``approximation.'' However, each of these arguments has a specialized, tailor-made flavor: Stone--Weierstrass-like arguments for metric spaces~\cite{KonigM18},  the unique structure theorem for analytic spaces~\cite{ClercFKP19}, and so on. It does not seem easy to distill  the essence that is common to different quantitative expressivity proofs. Indeed, there has not been a coalgebraic framework that unifies them.

\paragraph{Categorical Unification of Quantitative Expressivity via
 Codensity and Approximation}
We present the first categorical framework that uniformly axiomatizes  different approximation arguments---it uses a fibrational notion of \emph{approximating family}---and  unifies different quantitative expressivity results. 

Our framework hinges on  the construction called the \emph{codensity lifting}~\cite{SprungerKDH18,KKHKH19}; it is a general method for modeling a variety of bisimilarity-like notions (bisimilarity, probabilistic bisimilarity, bisimulation metric, etc.). The codensity lifting uses not only coalgebras (for unifying different state-transition systems) but also \emph{fibrations} for different \emph{observation modes}; the latter include Boolean predicates, quantitative/fuzzy predicates, equivalence relations, pseudometrics, topologies, etc. This use of fibrations provides flexibility to accommodate a variety of quantitative bisimilarity-like notions.

The codensity lifting, while defined in abstract categorical terms, has clear observational intuition (see~\cref{subsec:codensityBisim}). It also gives a class of \emph{codensity bisimilarity games} that characterize a variety of (qualitative and quantitative) bisimilarity notions~\cite{KKHKH19} (see also~\cref{subsec:codensityBisim}). 

Our key contribution of a categorical formalization of approximation is enabled by  the formalization of the codensity lifting. It has a similar observational intuition, too: see~\cref{subsec:approximableFamily}, where we characterize an approximating family of observations as a ``winnable'' set of moves in a suitable sense.

On top of our fibrational notion of approximating family, we establish a general expressivity framework, which is the first to unify existing quantitative expressivity results including~\cite{ClercFKP19,WildSP018,KonigM18}. In our unified framework, we have two proof principles for expressivity---Knaster--Tarski (\cref{thm:knasterTarskiApproxPrinciple})  and Kleene (\cref{thm:kleeneApproxPrinciple})---that mirror two classic characterizations of greatest fixed points. 
Our general framework is presented in terms of predicate lifting~\cite{schr08:expr,Pattinson04}. This is mostly for presentation purposes (showing concrete syntax is easier this way). A more abstract and fully fibrational recap of our framework---where a modal logic is formalized with a dual adjunction~\cite{BonsangueK05,BonsangueK06,PavlovicMW06,Klin07}---is found in~\cref{sec:dualityAndComma}.

We demonstrate our general framework with three examples: expressivity for the Kantorovich bisimulation metric (from~\cite{KonigM18}, \cref{sec:exampleKoenigMM}); that for Markov process bisimilarity (from~\cite{ClercFKP19}, \cref{sec:exMarkovProcesses}); and that for the so-called \emph{bisimulation uniformity} (\cref{sec:exBisimUnif}).
Both the Knaster--Tarski and Kleene principles are used for proofs. See~\cref{table:introParametersX}.
The last
is a new expressivity result that is not previously found in the literature. 

We note that the role of the notion of approximating family is as a useful axiomatization: it tells us what key lemma to prove in an expressivity proof, but it does not tell how to prove the key lemma. The proof of this key lemma is where the technical hardcore lies in existing expressivity proofs (by a Stone--Weierstrass-like result in~\cite{KonigM18},  by the unique structure theorem in~\cite{ClercFKP19}, etc.). For the new instance of bisimulation uniformity (\cref{sec:exBisimUnif}), the general axiomatization of approximating family allowed us to discover a result we need in a paper~\cite{CsaszarApproximation1971} that is seemingly unrelated to modal logic. The same result guided us in the design of modal logic, too, especially in the choice of propositional connectives.

\paragraph{Contributions} We summarize our contributions.
\begin{itemize}
 \item  The notion of approximating family, whose instances occur in the key steps of existing quantitative expressivity proofs. It is built on top of the codensity lifting. 
 \item We use it in a unified categorical  expressivity framework. It offers two proof principes (Knaster--Tarski and Kleene) that have different applicability (\cref{table:introParametersX}).
 \item The framework is instantiated to two known expressivity results~\cite{KonigM18,ClercFKP19} and one new result (\cref{sec:exBisimUnif}).
 \item The framework is given an abstract and fully fibrational recap (\cref{sec:dualityAndComma}) that exposes further fibrational structures. 
\end{itemize}

\begin{table*}[tbp]
	\caption{Examples of expressivity situations. The non-shaded rows describe data in expressivity situations, and the shaded ones describe the resulting bisimilarity notions and modal logics, i.e.\ two constructs compared in the problem of expressivity.
	}
	\label{table:introParametersX}
	\footnotesize
	\centering
	\scalebox{1}{\begin{tabular}{c||c|c|c}
			parameter
&
			\cite{KonigM18} (\cref{sec:exampleKoenigMM}) proved by Kleene			&
			\cite{ClercFKP19} (\cref{sec:exMarkovProcesses}) proved by Knaster--Tarski
			&
			\cref{sec:exBisimUnif} proved by Knaster--Tarski
			\\\hline\hline
			$\Cat{C}$
			&
			$\Set$
			&
			$\Meas$
			&
			$\Set$
			\\ 
			category of \emph{spaces}
			&
			sets 
			&
			measurable sets
			&
			sets
			\\\hline
			$B$
			&
			$B$ (arbitrary)
			&
			$(\SGiry \place)^A$
			&
			$B$ (arbitrary)
			\\
behavior functor
			&
			
			&
continuous-space Markov processes
			&
			\\\hline
			$p\colon \Cat{E}\to\Cat{C}$
			&
			$\PMetb\to\Set$ 
			&
			$\EqRelMeas\to\Meas$
			&
			$\Unif\to\Set$
			\\
observation mode
			&
1-bounded pseudometrics
			  &
			equiv.\ relations
&
			uniformity
			\\\hline
			\parbox[t]{8em}{\centering
				$\Omega\in\Cat{C}$\\
}
			&
			\parbox[t]{8em}{\centering
				$[0,1]$\\
}
			&
			\parbox[t]{13em}{\centering
				$2=\{0,1\}$ \\ 
}
			&
			\parbox[t]{8em}{\centering
				$[0,1]$\\
}
		       \\
				truth-value domain
		      &
				the unit interval
			  &
		      with the discrete $\sigma$-algebra
			      &
				the unit interval				  
			\\\hline
			\parbox[t]{9em}{\centering
				$\bOmega\in \Cat{E}_{\Omega}$\\
				observation predicate
			}
			&
			\parbox[t]{8em}{\centering
				$([0,1], d_e)$\\
				Euclidean metric
			}
			&
			\parbox[t]{8em}{\centering
				$(2,=)$\\
				equality
			}
			&
			\parbox[t]{8em}{\centering
				$([0,1], \mathcal{U}_e)$\\
				metric uniformity
			}
			\\\hline
			$(\tau_{\lambda}\colon B\Omega\to \Omega)_{\lambda\in\Lambda}$
			&
			arbitrary but
			&
			$\Lambda=A\times(\mathbb{Q}\cap[0,1])$
			&
			arbitrary but
			\\
observation modality
			&
			must satisfy \cref{asm:exampleKoenigMM}
			&
			$\tau_{a,r}((\mu_a)_{a\in A})=1$ iff $\mu_a(1)>r$
			&
			must satisfy \cref{asm:exBisimUnif}
			\\\hline\hline
			\cellcolor{lightgray!30}
			\parbox[t]{10em}{\centering
				\textbf{ resulting
					bisimilarity-like notion 
			}} 
			&
			\cellcolor{lightgray!30}
			\parbox[t]{7em}{\centering
				$B$-bisimulation metric~\cite{KonigM18}
			}
			&
			\cellcolor{lightgray!30}
			\parbox[t]{7em}{\centering
				probabilistic bisimilarity
			}
			&
			\cellcolor{lightgray!30}
			\parbox[t]{7em}{\centering
				bisimulation uniformity
			}
			\\\hline\hline
			\parbox[t]{11em}{\centering
				$\Sigma$\\
				propositional connectives
			}
			&
			\parbox[t]{10em}{\centering
				$\top,\neg,\min$, and\\ $(\ominus q)$ for $q\in\mathbb{Q}\cap[0,1]$
			}
			&
			\parbox[t]{7em}{\centering
				$\top,\wedge$
			}
			&
			\parbox[t]{10em}{\centering
				$1,\min$, and\\ $(r+),(r\times)$ for 
				$r\in\RR$
			}
			\\\hline
			\parbox[t]{10em}{\centering 
				$(f_\sigma)_{\sigma\in\Sigma}$\\
				propositional structure
			}
			&
			\parbox[t]{14em}{\centering
				Zadeh logic connectives on $[0,1]$
			}
			&
			\parbox[t]{7em}{\centering
				meet-semilattice with $0\sqsubseteq 1$
			}
			&
			\parbox[t]{7em}{\centering
				affine lattice structure on $\RR$
			}
			\\\hline\hline
			\bfseries 
			\cellcolor{lightgray!30}
			\parbox[t]{14em}{\centering
				resulting modal logic\\
				(modal operators are $(\heartsuit_{\lambda})_{\lambda\in\Lambda}$)
			}
			&
			\cellcolor{lightgray!30}
			\parbox[t]{15em}{\centering
				The logic in~\cite{KonigM18}, generalization of Zadeh fuzzy modal logic in~\cite{WildSP018}
			}
			&
			\cellcolor{lightgray!30}
			\parbox[t]{5em}{\centering
				$\mathrm{PML}_\wedge$~\cite{ClercFKP19}
			}
			&
			\cellcolor{lightgray!30}
			\parbox[t]{5em}{\centering
				A new modal logic
			}
		\end{tabular}
	}
\end{table*}

\paragraph{Related Work} 
Here we list related work considering quantitative expressivity.

Our framework is parameterized both in the kind of coalgebra and in the observation mode.
To our knowledge, the only existing work with this generality is~\cite{KR20}
which combines coalgebras and fibrations to provide a general setting
for proving expressivity. However, that approach does not accommodate approximation arguments, 
therefore failing to cover any of the aforementioned quantitative expressivity proofs~\cite{BreugelW05,ClercFKP19,WildSP018,KonigM18,WildS20}.
Our~\cref{sec:dualityAndComma} can be seen as an extension of~\cite{KR20}; our main novelty is the accommodation of approximation arguments and thus quantitative expressivity results, as we already discussed.

The idea of behavioral metrics was first proposed in~\cite{Giaca90}.
In the setting of category theory, the \emph{behavioral pseudometric} is introduced in~\cite{BreugelW05} in terms of coalgebras in the category $\PMetb$ of 1-bounded pseudometric spaces, and a corresponding expressivity
result is established. 
Many other formulations of quantitative bisimilarity are based on \emph{fibrational coinduction}~\cite{HJ98}. The work~\cite{BaldanBKK18} discusses general behavioral metrics (but not modal logics); expressivity w.r.t.\ these metrics is studied in~\cite{KonigM18} for general $\Set$-coalgebras. The line of work on \emph{codensity bisimilarity}---including \cite{SprungerKDH18,KKHKH19} and the current work---follows this fibrational tradition, too.

A recent work~\cite{WildS20}
uses a different  formulation of bisimilarity-like notions: it does not use fibrations or functor liftings, but uses so-called \emph{fuzzy lax extensions} of functors.
This approach is a descendant of \emph{relators}~\cite{Rutten98Relator}; seeking the connection to these works is future work.

\paragraph{Organization}
After recalling  preliminaries in~\cref{sec:prelim}, we axiomatize the data under which we study expressivity---it is called an  \emph{expressivity situation}---in~\cref{sec:expSit}.
In~\cref{sec:approxAndExp} we define our key notion of approximating family, from which we derive the Knaster--Tarski and Kleene proof principles for expressivity.
\S{}\ref{sec:exampleKoenigMM}--\ref{sec:exBisimUnif} present instances of our framework: two known~\cite{KonigM18,ClercFKP19} and one new (\cref{sec:exBisimUnif}).  In~\cref{sec:dualityAndComma}, we recap our framework in more abstract terms, and identify expressivity as a problem of comparing coinductive predicates in two different fibrations.

Most proofs are deferred to the appendix.

\section{Preliminaries}\label{sec:prelim}

We use coalgebras (\cref{subsec:prelimCoalg}) to accommodate different types of systems, and fibrations  (\cref{subsec:prelimFibration}) to accommodate different ``observation modes.'' Then quantitative bisimilarity-like notions are formulated as  coinductive predicates (\cref{subsec:prelimcoinductionInFib}).

\subsection{Coalgebra}
\label{subsec:prelimCoalg}
Coalgebras are commonly used as a categorical presentation of state-based transition systems~\cite{jacobs-coalg,Rutten00}. Let $\Cat{C}$ be a category and $B\colon \Cat{C}\to\Cat{C}$ be a functor. A \emph{$B$-coalgebra} is a pair $(X, x)$ of an object $X\in \Cat{C}$ and a $\Cat{C}$-arrow $x\colon X\to BX$; this coalgebra is often denoted simply by $x\colon X\to BX$. A \emph{morphism} of $B$-coalgebras from $(X, x)$ to $(Y,y)$ is a $\Cat{C}$-arrow $f\colon X\to Y$ such that $y\co f= Bf\co x$. 

The theory of coalgebras generalizes process theory and automata theory, where our interests are principally in the observational behaviors of transition systems that are insensitive to internal states. \emph{Bisimilarity} by Park and Milner~\cite{Milner89} is a prototype notion that captures such black-box behaviors. 
One of the early successes of coalgebras is a  categorical characterization of bisimilarity that works for different $B$'s, hence for different types of systems. In the theory, coalgebra morphisms are identified as ``behavior-preserving maps.''

A \emph{final} $B$-coalgebra is $\zeta\colon Z\to BZ$ such that there is a unique morphism from each coalgebra $(X,x)$ to $(Z,\zeta)$. It plays an important role 
as a fully abstract domain of $B$-behaviors. 

This paper's use of coalgebras goes beyond what we described so far. We follow~\cite{HJ98,BonchiPPR17,HasuoKC18,Komorida20} and use them in combination with \emph{fibrations} (\cref{subsec:prelimFibration}). In this  case, a coalgebra in a fiber is understood as a predicate (or relation, pseudometric, etc.) with a suitable invariance property (\cref{subsec:prelimcoinductionInFib}). The importance of final coalgebras remains, since those in a fiber---called coinductive predicates in~\cref{subsec:prelimcoinductionInFib}---characterize coinductively defined bisimilarity-like notions.

\subsection{Fibration}\label{subsec:prelimFibration}
Fibrations give a categorical way of organizing indexed entities. Roughly corresponding to indexed categories $\Op{\Cat{C}}\to \Cats$, a fibration $p\colon \Cat{E}\to \Cat{C}$ can be thought of as a collection $(\Cat{E}_{C})_{C\in \Cat{C}}$ of categories $\Cat{E}_{C}$ given for each $C\in \Cat{C}$, together with a suitable ``pullback'' action of arrows of $\Cat{C}$. In a fibration, these \emph{fiber categories} $\Cat{E}_{C}$ are patched up to form a single \emph{total category} $\Cat{E}$, a formulation that accommodates many structural reasoning principles. See~\cite{Jacobs:fib} for a comprehensive account.

The definition of fibration is simpler if restricted to posetal fibration; see e.g.~\cite{KKHKH19}. We need the following general definition for the discussions in~\cref{sec:dualityAndComma}.
\begin{definition}[fibration]\label{def:fibration}
	Let $p\colon \Cat{E}\to \Cat{C}$ be a functor. 
	\begin{itemize}
		\item An $\Cat{E}$-arrow $f\colon X\to Y$ is said to be \emph{Cartesian} if it has the following universal property: for 
each $\Cat{E}$-arrow $h\colon Z\to Y$, if $ph$ (in $\Cat{C}$) factors through $pf$ (say $ph=pf\co k$, see below middle), then there exists a unique $\Cat{E}$-arrow $g\colon Z\to X$ such that $h=f\co g$ and $pg = k$. 
\begin{displaymath}\footnotesize
			\vcenter{\xymatrix@R=0em@C-1em{
					\Cat{E} \ar[ddd]_-{p}
					& Z  
					\ar@/^/[rrd]^-{h}
					\ar@{-->}[rd]_-{g}
					\\
					&
					&
					X \ar[r]_-{f}
					&
					Y
					&
					i^{*}Y
					\ar[r]^-{\cart{i}Y}
					&
					Y
					\\
					& pZ  
					\ar@/^/[rrd]^(.7){ph}
					\ar@{->}[rd]_-{k}
					\\
					\Cat{C}
					&
					&
					pX \ar[r]_-{pf}
					&
					pY
					&
					C
					\ar[r]_-{i}
					&
					pY
			}}
		\end{displaymath}
\item The functor $p$ is called a \emph{fibration} if, for each $Y\in \Cat{E}$ and each $\Cat{C}$-arrow $i\colon C\to pY$, there exists an $\Cat{E}$-arrow $\cart{i}Y\colon i^{*}Y\to Y$ such that $p(\cart{i}Y)=i$ and $\cart{i}Y$ is Cartesian. See above right. 
The object $i^{*}Y\in \Cat{E}$ is called the \emph{pullback} of $Y$ along $i$ (it is unique up-to isomorphism); the arrow $\cart{i}Y$ is called the \emph{Cartesian lifting} of $i$ with respect to $Y$. 
	\end{itemize}
	We say that $X\in \Cat{E}$ is \emph{above} $C\in \Cat{C}$ if $pX=C$; an $\Cat{E}$-arrow \emph{above} a $\Cat{C}$-arrow is defined similarly. A fibration $p$ gives  rise to the \emph{fiber category} $\Cat{E}_{C}$ for each $C\in \Cat{C}$: it consists of all the objects above $C$ and all the arrows above $\id_{C}$. The category $\Cat{E}$ is called the \emph{total category} of the fibration. Each $\Cat{C}$-arrow $i\colon C\to C'$ gives rise to a \emph{reindexing functor} $i^{*}\colon \Cat{E}_{C'}\to \Cat{E}_{C}$. 
	
	A fibration $p\colon \Cat{E}\to \Cat{C}$ is \emph{posetal} if each fiber $\Cat{E}_{C}$ is a poset. It is a \emph{$\CLatw$-fibration} if each fiber is a complete lattice, and each reindexing functor preserves arbitrary meets $\bigsqcap$. 
\end{definition}
$\CLatw$-fibrations form a special class of \emph{topological functors}~\cite{herr74:topo}. 
We prefer the fibrational presentation, following  works on coinductive predicates~\cite{HJ98,BonchiPPR17,HasuoKC18,SprungerKDH18,KR20,Komorida20}.

Fibrations in general organize various indexed structures. In this paper, however, our examples share the following intuition.
\begin{itemize}
	\item The base category $\Cat{C}$ is that of sets (or spaces, structured sets, etc.) and functions between them (that preserve/respect those structures imposed on the objects). We often assume an endofunctor $B\colon \Cat{C}\to\Cat{C}$, for which coalgebras  model state-based systems, as in~\cref{subsec:prelimCoalg}. 
	\item A fibration $p\colon \Cat{E}\to\Cat{C}$ specifies the \emph{observation mode}, providing an additional reasoning structure for spaces in $\Cat{C}$. Such a structure can be predicates, binary relations, pseudometrics, topologies,  etc.\ (Example~\ref{ex:fibr}). In particular, a bisimilarity-like notion over $X$
	is  an object $P\in \Cat{E}_{X}$.
\end{itemize}

\noindent For the sake of presentation, we fix the following terminology. 
\begin{definition}[(fib.) predicate]\label{def:fibrationalPredicate}
	In a fibration $p\colon \Cat{E}\to\Cat{C}$, an object $P\in \Cat{E}_{X}$ is called a \emph{(fibrational) predicate} over $X$. 
\end{definition}
\noindent Note that a (fibrational) predicate can be in fact a binary relation, a pseudometric, etc., depending on the choice of $p\colon \Cat{E}\to\Cat{C}$. This abuse of words will be useful in~\cref{subsec:prelimcoinductionInFib}. 

The coming examples are all $\CLatw$-fibrations;  arrows in fibers are denoted by $\sqsubseteq$. The intuition of $P\sqsubseteq Q$ is that the predicate $P$ is more fine-grained and discriminating than $Q$. 

\begin{example}
	\label{ex:fibr}
	($\Pred\to\Set$) An object of $\Pred$ is a pair $(X,P)$ of a set $X$ and a predicate $P\subseteq X$. An arrow $f\colon (X,P)\to (Y, Q)$ in $\Pred$ is a function $f\colon X\to Y$ such that $f(x)\in Q$ whenever $x\in P$. The obvious forgetful functor $\Pred\to\Set$ is a fibration, with pullbacks given by $f^{*}(Y,Q)=(X, f^{-1}Q)$. Each fiber $\Pred_{X}$ is the powerset $\Pow X$ with ${\sqsubseteq}={\subseteq}$ as arrows. 
	
	($\ERel\to\Set, \EqRel\to\Set, \Pre\to\Set$) $\ERel\to\Set$ is a binary variant of $\Pred\to\Set$: an object $(X,R)$ of $\ERel$ is a set $X$ with an endorelation $R\subseteq X\times X$. Pullbacks are given by inverse images, too. These relations are restricted to equivalence relations and preorders,  in $\EqRel\to\Set$ and $\Pre\to\Set$, respectively.
	
	($\PMetb\to\Set$) $\PMetb$ consists of sets with 1-bounded pseudometrics and non-expansive maps. A pseudometric $d$ is much like a metric but allows $d(x,y)=0$ for $x\neq y$, a common setting where bisimulation metrics are formulated. 1-boundedness (that $d(x,y)\le 1$ for all $x,y$) is assumed for technical convenience---any bound would do, such as $\infty$.  The forgetful functor $\PMetb\to\Set$ is a fibration, in which pullbacks equip a set with an induced pseudometric: $f^{*}(Y,e)=\bigl(X, \lambda (x,x').\, e(f(x), f(x'))\bigr)$. A consequence is that the fiber $(\PMetb)_{X}$ consists of all 1-bounded pseudometrics over $X$ ordered by ${\sqsubseteq} = {\ge}$---this concurs with the above intuition that $d$ is more discriminating if $d\sqsubseteq d'$. 
	
	($\EqRelMeas\to\Meas$)
This is a measurable variant of $\EqRel\to\Set$; $\Meas$ is the category of measurable spaces and measurable maps. An object $(X, R)$ of $\EqRelMeas$ is given by a measurable set $X$, together with $R\subseteq X\times X$.
	\begin{auxproof}
		such that, if $(x,y)\notin R$, then there exists some measurable $S\subseteq X$ satisfying $x\in S$ and $y\notin S$. (There is some subtlety here: an alternative is to ask $R$ to be measurable in $X\times X$, but this  does not give a $\CLatw$-fibration.)
	\end{auxproof}
	\begin{auxproof}
		Future work: give a nice categorical characterization to $\EqRelMeas\to\Meas$. Is it a pullback of $\mathrm{StrongSub}(\Meas)\to\Meas$?
	\end{auxproof}
	
	\begin{auxproof}
		($\Top\to\Set, \Meas\to\Set$) The forgetful functor $\Top\to\Set$ is a fibration, where reindexing is given by pullback topologies. The fiber $\Top_{X}$ has topologies $\mathcal{O}\subseteq \Pow X$ as objects, ordered by ${\sqsubseteq} = {\supseteq}$. $\Meas\to\Set$ is similar; its fiber $\Meas_{X}$ consists of all $\sigma$-algebras over $X$. 
	\end{auxproof}
\end{example}

\subsection{Coinductive Predicate in a Fibration}\label{subsec:prelimcoinductionInFib}
The combination of coalgebras and fibrations has been actively studied, starting in~\cite{HJ98} and more recently e.g.\ in~\cite{BonchiPPR17,HasuoKC18,KKHKH19,KR20,Komorida20}. Fibrations introduce additional reasoning structures  (\cref{subsec:prelimFibration})
which allow to accommodate bisimilarity-like notions beyond classical bisimilarity (including bisimilarity pseudometrics, see e.g.~\cite{Bonchi0P18,KKHKH19}). In fact, these bisimilarity-like notions are defined coinductively, i.e., as suitable greatest fixed points, which are then identified with coalgebras in a fiber. Here we shall review this fibrational machinery for coinductive reasoning for coalgebras.

\begin{definition}[functor lifting]\label{def:fibrationalFunctorLifting}
	Let $p\colon \Cat{E}\to \Cat{C}$ be a fibration, and $B\colon \Cat{C}\to\Cat{C}$ be a functor. A functor $\ol B\colon \Cat{E}\to \Cat{E}$ is  a \emph{lifting} of $B$ along $p$ if $p\co \ol B=B\co p$. 
We say that a lifting $\ol B$ is  \emph{fibered} if it preserves Cartesian arrows.
	\begin{auxproof}
		, in the sense that the canonical arrow below (dashed) is an iso.
		\begin{equation}\label{eq:fiberedlifting}
			\vcenter{\xymatrix@R=0em@C-.6em{
					&&&&
					(Bf)^{*}(\ol{B}Q)
					\ar[rd]^-{\cart{Bf}(\ol{B}Q)}
					\\
					\Cat{E}
					\ar[r]^-{\ol B}
					\ar[ddd]_-{p}
					&
					\Cat{E}
					\ar[ddd]_-{p}
					& 
					f^{*}Q
					\ar[r]^-{\cart{f}Q}
					&
					Q
					&&
					{\ol B}Q
					\\
					&&&&
					{\ol B}(f^{*}Q)
					\ar[ru]_-{\ol{B}(\cart{f}Q)}  
					\ar@{-->}[uu]_-{\cong}
					\\
					\\
					\Cat{C}
					\ar[r]^-{B}
					&
					\Cat{C}
					&
					X \ar[r]^-{f}
					&
					Y
					&
					BX \ar[r]^-{Bf}
					&
					BY
			}}
		\end{equation}
	\end{auxproof}
\end{definition}
A functor lifting determines the type of \emph{coinductive predicates} through a \emph{predicate transformer}.

\begin{definition}[(fibrational) predicate transformer $x^{*} \circ \ol B$]\label{def:fibrationalPredTrans}
	In the setting of \cref{def:fibrationalFunctorLifting}, let $\ol B$ be a lifting of $B$, and  $x\colon X\to BX$ be a $B$-coalgebra. We then obtain an endofunctor 
	\begin{displaymath}
		x^{*}\co\ol B\colon \Cat{E}_{X}\to \Cat{E}_{X}
		~\text{by the composite}~
		\Cat{E}_{X}\xrightarrow{\ol B}\Cat{E}_{BX}\xrightarrow{x^{*}}\Cat{E}_{X}.
	\end{displaymath}
	The functor $x^{*}\co\ol B$ is called the \emph{predicate transformer} induced by $\ol{B}$ over the $B$-coalgebra $x$.
\end{definition}

It is standard to characterize bisimilarity as a suitable greatest fixed point (gfp). Accordingly, we are interested in the gfp of the predicate transformer $x^{*}\co\ol B$. The latter amounts to the final $x^{*}\co\ol B$-coalgebra,  in our current setting where the fiber $\Cat{E}_{X}$ is not necessarily a poset but a category. 

\begin{definition}[coinductive predicate $\nu(x^{*}\co\ol{B})\in\Cat{E}_{X}$, and invariant]\label{def:fibrationalInvariantAndCoinductivePred}
	In the setting of \cref{def:fibrationalPredTrans}, the carrier of the  final $x^{*}\co\ol{B}$-coalgebra (if it exists) is called the \emph{$\ol{B}$-coinductive predicate}  over $x$. It is denoted by $\nu(x^{*}\co\ol{B})\in \Cat{E}_{X}$. 
	
	An $x^{*}\co\ol{B}$-coalgebra is called a \emph{$\ol{B}$-invariant} over $x$. 
\end{definition}
\noindent The names in the above definition reflect the common reasoning principle for gfp specifications (such as safety), namely that an invariant underapproximates (and thus witnesses) the gfp specification. Each $\ol{B}$-invariant indeed witnesses the $\ol{B}$-coinductive predicate, in the sense that there is a unique morphism from the former to the latter. 

\begin{example}
	[$\Pred\to\Set$] Liftings of a functor $B$ along $\Pred\to\Set$ are well-studied---they correspond to the so-called \emph{predicate liftings}~\cite{schr08:expr,Pattinson04}. For example, for the powerset functor $B=\Pow$, two liftings $\ol{\Pow}_{\Box}, \ol{\Pow}_{\Diamond}$ are given by 
	\begin{align*}
\ol{\Pow}_{\Box}(X,P) & =
		\bigl(\,\Pow X,\,\Box P = \{U\subseteq X\mid U\subseteq P\} \,\bigr),
                                                \\
		\ol{\Pow}_{\Diamond}(X,P) & =
		\bigl(\,\Pow X,\,\Diamond P = \{U\subseteq X\mid U\cap P\neq\emptyset\} \,\bigr).
	\end{align*}
Thus a choice of  lifting here amounts to a choice of modality.
	
	Coinductive predicates in this setting represent various safety properties. On a Kripke frame $x\colon X\to \Pow X$, 
the coinductive  predicate $\nu(x^{*}\co\ol{\Pow}_{\Diamond})\subseteq X$ designates those states from which there is an infinite path. A similar example appears in~\cite{BonchiPPR17}. \end{example}

\begin{example}
	[$\ERel\to\Set, \EqRel\to\Set, \Pre\to\Set $] In these relational examples, a coinductive predicate embodies some bisimilarity-like relation---more specifically, the greatest among those relations which are preserved by one-step  transitions. A class of examples is given by \emph{coalgebraic bisimilarity}: see~\cite{HJ98}, where they assign a specific choice of lifting $\ol{B}$ to each functor $B$. Other examples include similarity~\cite{HughesJ04} (see also \cite[\S{}4.3]{KR20}), probabilistic bisimilarity~\cite{KKHKH19} and the language equivalence between deterministic automata~\cite{KKHKH19}. 
\end{example}  
\begin{example}
	[$\PMetb\to\Set$] A prototypical example is given as follows. Let $B=\SDist$, the \emph{subdistribution} functor that carries a set $X$ to $\SDist X=\{\xi\colon X\to [0,1]\mid \sum_{x\in X}\xi(x)\le 1\}$. Its action on arrows is given by push-forward distributions. One lifting $\ol{\SDist}$ of $\SDist$ along $\PMetb\to\Set$ is given by the \emph{Kantorovich metric}:
	\begin{math}
		\ol{\SDist}(X,d)= (\SDist X, \mathcal{K}d)
	\end{math}, where $\mathcal{K}d(\xi, \xi')$ is given by
	\begin{align*}
\textstyle
		\sup_{f\colon (X,d)\to_{\text{ne}} [0,1]} 
		\, \bigl|\,\sum_{x\in X}f(x)\cdot \xi(x) - \sum_{x\in X}f(x)\cdot \xi'(x)\,\bigr|.
	\end{align*}
	In the above supremum, $f$ ranges over all \emph{nonexpansive} functions of the designated type. The coinductive predicate for $\ol{\SDist}$ coincides with the \emph{bisimulation metric} from~\cite{DesharnaisGJP04}. 
	
	The above construction of a lifting $\ol{B}$ along $\PMetb\to\Set$ has been generalized to an arbitrary functor $B$.  This is called the \emph{Kantorovich lifting} and is introduced in~\cite{BaldanBKK18}. It uses a map $B[0,1]\to[0,1]$ as a parameter; the latter is much like a choice of a modality. We study this general setting in~\cref{sec:exampleKoenigMM}, following~\cite{KonigM18} but identifying the Kantorovich lifting as a special case of the \emph{codensity lifting} (\cref{def:codensityLifting}). 
\end{example}

Here is an abstract account of coinductive predicates.
\begin{proposition}[from~{\cite[Prop.~4.1 and 4.2]{HasuoKC18}}]\label{prop:coindPredMoreAbstractly}
Each lifting $\ol B$  of $B$ along $p\colon \Cat{E}\to\Cat{C}$ induces a functor $\CA p\colon \CA{\ol{B}}\to \CA B$; it carries a coalgebra $U\to \ol{B}U$ (in $\Cat{E}$) to $pU\to B(pU)$ (in $\Cat{C}$). Moreover, the functor $\CA p$ is a fibration if $\ol B$ is fibered. The fiber $\CA{\ol{B}}_{(X,x)}$ over a coalgebra $x\colon X\to BX$ coincides with the category $\CA{x^{*}\co\ol{B}}$ of $\ol{B}$-invariants over $(X,x)$. \myqed
\end{proposition}

\begin{lemma}[from~{\cite{KKHKH19,HasuoKC18}}]\label{lem:CLatWFibCoindPred}
	Let $\ol{B}$ be a lifting of $B$ along $p$. If $p$ is a $\CLatw$-fibration, then $\nu(x^{*}\co\ol{B})$ exists for each $B$-coalgebra $x:X\to BX$---it is the gfp of the monotone map $x^{*}\co\ol{B}$ over the complete lattice $\Cat{E}_{X}$. Moreover, if $\ol{B}$ is fibered, then these coinductive predicates are preserved by reindexing along coalgebra morphisms. \myqed
\end{lemma}

\begin{notation}
	Final objects shall be denoted with subscripts ($1_{\Cat{E}}, 1_{X}$, etc.) to clarify in which category it is final, unless it is obvious. We will also write $\top$ for the maximum element of a poset, that is, the final object when the poset is thought of as a category. This typically happens 
	with the final object $\top_{X}\in \Cat{E}_{X}$ in a fiber $\Cat{E}_{X}$ of a $\CLatw$-fibration.
\end{notation}

\section{Expressivity Situation}
\label{sec:expSit}

On top of the above preliminaries, we fix the format of categorical data under which we study expressivity. It is called an expressivity situation.
While it may seem overwhelming, we  show that the data arises naturally, with clear intuition from the viewpoints of modal logics and observations (\S{}\ref{subsec:syntaxSemOfLogic}--\ref{subsec:codensityBisim}).

\subsection{Definition}\label{subsec:expSitDef}

\begin{definition}\label{def:expSitX}
	An \emph{expressivity situation} 
 $\esit=(p,B,\Omega,\bOmega,\Sigma,\Lambda,(f_\sigma)_{\sigma\in\Sigma},(\tau_\lambda)_{\lambda\in\Lambda})$ is given by the following.
	\begin{itemize}
		\item A $\CLatw$-fibration $p\colon\Cat{E}\to\Cat{C}$.
		\item A functor $B\colon\Cat{C}\to\Cat{C}$ (a \emph{behavior functor}).
		\item An object $\Omega\in\Cat{C}$ (a \emph{truth-value object}) equipped with finite powers ($\Omega^{n}\in\Cat{C}$ for $n\in \mathbb{N}$), and another object $\bOmega$ (an \emph{observation predicate}) above it. It follows that $\bOmega$ also has finite powers~\cite[Prop.~9.2.1]{Jacobs:fib}.
		\item A ranked alphabet $\Sigma$ of \emph{propositional connectives} and a family of arrows $\left(f_\sigma\colon\Omega^{\rank(\sigma)}\to\Omega\right)_{\sigma\in\Sigma}$ (a \emph{propositional structure}). Moreover, we require that each $f_\sigma\colon\Omega^{\rank(\sigma)}\to\Omega$ has a \emph{lifting} $g_\sigma\colon\bOmega^{\rank(\sigma)}\to\bOmega$ (in $\Cat{E}$) such that $pg_\sigma=f_\sigma$.
\item A set $\Lambda$ of \emph{modality indices} and a family of algebras $\left(\tau_\lambda\colon B\Omega\to\Omega\right)_{\lambda\in\Lambda}$ (\emph{observation modalities}).
	\end{itemize}
\end{definition}

Roughly speaking, $\Sigma$ and $\Lambda$ are used for modal logic \emph{syntax}, and $\Cat{C}$, $B$, $\Omega$, $(f_\sigma)_{\sigma\in\Sigma}$, and $(\tau_\lambda)_{\lambda\in\Lambda}$ are used for modal logic \emph{semantics}.
The other constructs ($p$ and $\bOmega$) are there for defining a bisimilarity-like notion.

In what follows, we formulate the expressivity problem on top of~\cref{def:expSitX}, explaining the role of each piece of data in an expressivity situation $\esit$. More specifically, we let $\esit$ induce the following constructs: 1) the modal logic  $L_\esit$ (\cref{def:syntax,def:formulaSemantics}); 2) the \emph{fibrational logical equivalence} $\LEq{\esit}{x}$ induced by $L_\esit$  (\cref{def:logicalEq}); and 3) the bisimilarity-like notion
 $\CBisim{\bOmega}{\tau}{x}$ as a \emph{codensity bisimilarity} (\cref{def:codensityBisim}). Comparison of the last two is the problem of expressivity.
As  an illustrating example, we use an expressivity situation $\esit_{\mathrm{KMM}}$ that arises from the real-valued logic $\mathcal{M}(\Lambda)$ in~\cite{KonigM18} (see also~\cref{sec:exampleKoenigMM}).

\subsection{Syntax and Semantics of Our Logic $L_\esit$}\label{subsec:syntaxSemOfLogic}

The syntax of modal logic is specified by the propositional connectives in $\Sigma$ and the modality indices in $\Lambda$.
\begin{definition}[$L_\esit$]
	\label{def:syntax}
	Let $\esit$ be an expressivity situation in \cref{def:expSitX}. The modal logic $L_{\esit}$ has the following syntax.
	\begin{align*}
		\varphi, \varphi_{1},\dotsc,  \varphi_{n}\;::=&\;\sigma(\varphi_{1},\dotsc,\varphi_{\rank(\sigma)})&(\sigma\in\Sigma) \\
		\mid&\;\heartsuit_{\lambda}\varphi&(\lambda\in\Lambda)
	\end{align*}
	We also let $L_{\esit}$ denote the set of all formulas.
\end{definition}

\begin{example}
	\label{ex:syntaxKMM}
	Let $\Lambda$ be a set.
	To model the modal logic $\mathcal{M}(\Lambda)$ in~\cite{KonigM18}, we let $\Sigma=\{\top^0,\min^2,\neg^1\}\cup\{(\ominus q)^1~|~q\in\mathbb Q\cap[0,1]\}$.
Then the syntax is given by
	\begin{align*}
		\varphi, \varphi_{1},\dotsc,  \varphi_{n}\;::=&\;\top\mid\neg(\varphi)\mid\min(\varphi_{1},\varphi_{2})&\\
		\mid&\;(\ominus q)\varphi\enspace(q\in\mathbb Q\cap[0,1])\;\mid\;\heartsuit_{\lambda}\varphi\enspace(\lambda\in\Lambda). 
	\end{align*}
	 Identifying $\heartsuit_{\lambda}$ with $[\lambda]$ in the original notation, this recovers the syntax of $\mathcal{M}(\Lambda)$ in~\cite{KonigM18}.
\end{example}

Given a coalgebra $x\colon X\to BX$ of the behavior functor $B$, the semantics of each formula is a $\Cat{C}$-arrow from the state space $X$ to the truth-value object $\Omega$, inductively defined as follows.
\begin{definition}
	\label{def:formulaSemantics}
	Let $\esit$ be an expressivity situation in \cref{def:expSitX} and let $x\colon X\to BX$ be a $B$-coalgebra.
	For each $\varphi\in L_{\esit}$, the \emph{interpretation}  $\sem{\varphi}_{x}\colon X\to\Omega$ of $\varphi$ with respect to $x$ is defined inductively as follows:
	\begin{align*}
		\sem{\sigma(\varphi_{1},\dotsc,\varphi_{\rank(\sigma)})}&=f_{\sigma}\co\langle\sem{\varphi_{1}},\dotsc,\sem{\varphi_{\rank(\sigma)}}\rangle,  &(\sigma\in\Sigma) \\
		\sem{\heartsuit_{\lambda}\varphi}&=\tau_{\lambda}\co(B\sem{\varphi})\co x.&(\lambda\in\Lambda)
	\end{align*}
\end{definition}

\begin{example}
	\label{ex:formulaSemanticsKMM}
	Recall \cref{ex:syntaxKMM}.
	Let $B\colon\Set\to\Set$ be an endofunctor, and
$\Omega$ be the unit interval $[0,1]$.
	We specify the propositional structure $(f_\sigma\colon[0,1]^{\rank(\sigma)}\to[0,1])_{\sigma\in\Sigma}$ by:
	\begin{align*}
		f_{\top}()&=1,&f_{\min}(x,y)&=\min(x,y),\\
		f_{\neg}(x)&=1-x,&f_{\ominus q}(x)&=\max(x-q,0).
	\end{align*}
	Here $\min$ plays the role of conjunction.
	Let $(\tau_\lambda\colon B[0,1]\to[0,1])$ be  a family of observation modalities, and  $x\colon X\to BX$ be a $B$-coalgebra.
	Then, 
the semantics $\sem{\varphi}_x$ of each formula $\varphi$  in \cref{def:formulaSemantics} coincides with the definition in~\cite[\S{}3.2]{KonigM18}.
\end{example}

The following definition  generalizes, in fibrational terms, the conventional definition that two states are logically equivalent if each formula's truth values coincide.

\begin{definition}[fibrational logical equivalence $\LEq{\esit}{x}$]
	\label{def:logicalEq}
	Let $\esit$ be an expressivity situation in \cref{def:expSitX} and let $x\colon X\to BX$ be a $B$-coalgebra.
	The \emph{fibrational logical equivalence} $\LEq{\esit}{x}$ with respect to $x$ is a predicate above $X$ defined by \[
		\LEq{\esit}{x}=\textstyle\bigsqcap_{\varphi\in L_{\esit}}\sem{\varphi}_{x}^{*}\bOmega, \quad\text{where }
\vcenter{		\xymatrix@R=.5em@C-.5em{
			\Cat{E} \ar[d]^-p & \sem{\varphi}_{x}^{*}\bOmega \ar@{-->}[r]
			& \bOmega
			\\ 
\Cat{C} &   X\ar[r]^-{\sem{\varphi}_{x}} & \Omega\mathrlap{.}
		}
}
	\]
\end{definition}

\begin{example}
	\label{ex:logicalEqKMM}
	Recall \cref{ex:formulaSemanticsKMM}.
	To define a logical distance, we let $p$ be the $\CLatw$-fibration $\PMetb\to\Set$ (\cref{ex:fibr}), and $\bOmega$ be the usual Euclidean metric $d_e$ on $[0,1]$.
	
	Then, for each $B$-coalgebra $x\colon X\to BX$ , the pseudometric $d^L_x := \LEq{\esit_{\mathrm{KMM}}}{x}$ is equivalently described by  \[
		d^L_x(s,t) = \textstyle\sup_{\varphi} d_e\bigl(\sem{\varphi}_x(s),\sem{\varphi}_x(t)\bigr),
	\] where $\varphi$ ranges over the modal formulas.
	Thus \cref{def:logicalEq} coincides with the notion of logical distance in~\cite[Def.~25]{KonigM18}.
\end{example}

\subsection{Codensity Lifting and Codensity Bisimilarity}\label{subsec:codensityBisim}
We unify different quantitative bisimilarity notions---such as probabilistic bisimilarity and bisimulation metric---using \emph{codensity bisimilarity}.  This is what is compared with the fibrational logical equivalence (\cref{def:logicalEq}). 
Codensity bisimilarity  arises natually from the notion of \emph{codensity lifting}~\cite{KKHKH19,SprungerKDH18}.

\begin{wrapfigure}[3]{r}{0pt}

\vspace{-.4em}
\begin{math}
  \vcenter{\xymatrix@R=.8em@C=.9em{
  \Cat{E}
      \ar[d]_-{p}
  &
  {k^{*}\bOmega}
     \ar[r]&
  {\bOmega}
 \\
  \Cat{C}
  &
  {X
}
    \ar[r]^{k
}
  &
  {\Omega}
 }} 
\end{math}
\end{wrapfigure}
The codensity lifting, although it is formulated in abstract terms (\cref{def:codensityLifting}), has clear  \emph{observational} intuition. The codensity lifting $\codlif{B}{\bOmega}{\tau}\colon\Cat{E}\to\Cat{E}$ is defined for the following data, which is part of the data in an expressivity situation (\cref{def:expSitX}).
\begin{itemize}
 \item A fibration $p\colon\Cat{E}\to\Cat{C}$
for the observation mode (\cref{subsec:prelimFibration}).
 \item An endofunctor $B\colon\Cat{C}\to\Cat{C}$; which is to be lifted (\cref{def:fibrationalFunctorLifting}). Our target system is a $B$-coalgebra (\cref{subsec:prelimCoalg}). 
\item A truth value domain  $\Omega\in\Cat{C}$. Here we use ``$\Omega$-valued observations,'' that is, arrows $k\colon X\to \Omega$ in $\Cat{C}$ (cf.\ the diagram above). 

 \item An observation predicate $\bOmega\in\Cat{E}$ above $\Omega$. It is a ``template of observations,'' whose pullback $k^{*}(\bOmega)$ by an observation $k$ defines an indistinguishability notion on $X$. See above.

For $\Cat{E}=\EqRel$, a common example is $\Omega=2$ and $\bOmega=(2,\Eq_2\subseteq 2\times 2)$; it means we distinguish elements of $X$ if they are mapped to different elements by an observation $k\colon X\to 2$.
 \item A family of observation modalities $(\tau_\lambda\colon B\Omega\to\Omega)_{\lambda\in\Lambda}$. 
An observation modality $\tau_{\lambda}$ specifies how observations interact with the behavior type $B$. Technically, it lifts
\begin{itemize}
 \item 
       an observation $k\colon X\to \Omega$ of $X$ 
 \item 
to an observation \begin{math}
 BX\xrightarrow{Bk} B\Omega\xrightarrow{\tau_{\lambda}}\Omega
\end{math}
 of $BX$.
\end{itemize}
\end{itemize}
Given the above data with the observational intuition, the codensity lifting is defined as the ``indistinguishability with respect to lifted $\Omega$-valued observations,'' as below.

\begin{definition}[codensity lifting]
	\label{def:codensityLifting}
	Let $\esit$ be an expressivity situation in \cref{def:expSitX}.
	The \emph{codensity lifting} of $B$ with respect to $\bOmega$ and $(\tau_\lambda)_{\lambda\in\Lambda}$ is the functor $\codlif{B}{\bOmega}{\tau}\colon\Cat{E}\to\Cat{E}$, defined by
	\begin{equation}\label{eq:codensity}
		\codlif{B}{\bOmega}{\tau}P = \textstyle
		\bigsqcap_{\lambda\in\Lambda,h\in\Cat{E}(P,\bOmega)}(\tau_\lambda\circ B(ph))^* \bOmega.
	\end{equation}
\end{definition}

\begin{wrapfigure}[4]{r}{0pt}

\vspace{-.4em}
\begin{math}
  \vcenter{\xymatrix@R=.8em@C=1.1em{
  \Cat{E}
      \ar[d]_-{p}
  &
  {\bullet}
     \ar@{-->}[rr]&&
  {\bOmega}
 \\
  \Cat{C}
  &
  {BX
}
    \ar[r]^{B(ph)
}
  &
  {B\Omega}
    \ar[r]^-{\tau_{\lambda}}
  &
  {\Omega}
 }} 
\end{math}
\end{wrapfigure}
Some explanations are in order.
An arrow $h\colon P\to \bOmega$ is a \emph{legitimate} $\Omega$-valued observation---it is a function $ph\colon X\to \Omega$ that respects indistinguishability predicates $P$ and  $\bOmega$. The latter legitimacy requirement instantiates to predicate- and relation-preservation, nonexpansiveness, continuity, etc., depending on the choice of the fibration $p$. 

The predicate $(\tau_\lambda\circ B(ph))^* \bOmega$ in~\cref{eq:codensity} is over $BX$ and induced from ``lifting  $h$ along $\tau_{\lambda}$,'' as shown in the above diagram. In~\cref{eq:codensity}, in the end, $\codlif{B}{\bOmega}{\tau}P$ arises as the coarsest indistinguishability induced by such $h$'s.

\begin{definition}[codensity bisimilarity $\CBisim{\bOmega}{\tau}{x}$]
	\label{def:codensityBisim}
	Let $\esit$ be an expressivity situation in \cref{def:expSitX} and let $x\colon X\to BX$ be a $B$-coalgebra.
	The \emph{codensity bisimilarity} $\CBisim{\bOmega}{\tau}{x}$ of $x$ is the $\codlif{B}{\bOmega}{\tau}$-coinductive predicate (\cref{def:fibrationalInvariantAndCoinductivePred}), i.e., the greatest fixed point of the map $x^*\co\codlif{B}{\bOmega}{\tau}\colon\Cat{E}_X\to\Cat{E}_X$: \[
		\CBisim{\bOmega}{\tau}{x}=\nu(x^*\co\codlif{B}{\bOmega}{\tau})\in\Cat{E}_X.
	\]
\end{definition}

\begin{example}
	\label{ex:codensityBisimKMM}
	Recall \cref{ex:logicalEqKMM}.
	In this case we can see that the codensity lifting coincides with the \emph{Kantorovich lifting};  see \cref{sec:exampleKoenigMM} for details.
	Thus the codensity bisimilarity coincides with the behavioral distance defined in~\cite[Def.~22]{KonigM18}.
\end{example}

Much like the codensity lifting  explained in terms of ``observations,'' codensity bisimilarity can be regarded as an outcome of ``repeated observations.'' We recall here the following game-theoretic characterization~\cite{KKHKH19}, which is not used in the rest of the paper yet is useful for providing intuitions.

\begin{fact}[{from~\cite[Cor.~VI.4]{KKHKH19}}]
	\label{fact:codensityGame}
	In the setting of \cref{def:codensityBisim}, we define a two-player infinite game called the \emph{(untrimmed) codensity bisimilarity game} as shown in Table~\ref{table:untrimmedCodensityGame}. Its two players are called Spoiler and Duplicator; once a player gets stuck,  the player loses;  any infinite play is won by Duplicator.

	Then, for each $P\in\Cat{E}_X$, $P$ is below the codensity bisimilarity ($P\sqsubseteq\CBisim{\bOmega}{\tau}{x}$) if and only if the position $P\in\Cat{E}_X$ is winning for Duplicator.
\end{fact}
\begin{table}[tbp]
  \caption{(Untrimmed) codensity bisimilarity game}
  \label{table:untrimmedCodensityGame}
  \centering \renewcommand{\arraystretch}{1}
  \begin{tabular}{l|l|l}
    position & player &  possible moves    \\\hline\hline
$P\in\Cat{E}_X$
             & Spoiler &  $k\colon X\to\Omega$ in $\Cat{C}$ such that   
    \\
             & &  $\exists\lambda\in\Lambda.\,P\not\sqsubseteq x^*(Bk)^*\tau_\lambda^*\bOmega$
    \\\hline
    $k\in\Cat{C}(X,\Omega)$
             & Duplicator &  $P\in\Cat{E}_X$ s.t.\ $P\not\sqsubseteq k^*\bOmega$
\end{tabular}
\end{table}

\noindent
In the codensity bisimilarity game, Spoiler repeatedly carries out observations $k\colon X\to \Omega$, trying to show that the previous move $P\in\Cat{E}_{X}$ by Duplicator was in fact not below the codensity bisimilarity.
Duplicator responds with a counter-argument that Spoiler's $k\colon X\to \Omega$ is illegitimate, not respecting the indistinguishability predicates $P$ (on $X$) and $\bOmega$ (on $\Omega$).

\subsection{Adequacy and Expressivity}\label{subsec:adequacyExpressivity}
 We are ready to formulate  adequacy and expressivity. Recall that  $P\sqsubseteq Q$  in a fiber means that $P$ is more discriminating.
\begin{definition}\label{def:expressive}
	Let $\esit$ be an expressivity situation (\cref{def:expSitX}) and  $x\colon X\to BX$ be a $B$-coalgebra.
\begin{itemize}
 \item 	$\esit$ is \emph{expressive} for $x$ if $\CBisim{\bOmega}{\tau}{x}\sqsupseteq\LEq{\esit}{x}$ holds.
 \item 	$\esit$ is \emph{adequate} for $x$ if $\CBisim{\bOmega}{\tau}{x}\sqsubseteq\LEq{\esit}{x}$ holds. 
\end{itemize}	
	$\esit$ is \emph{expressive} (or \emph{adequate}) if it is expressive (or adequate, respectively) for any $B$-coalgebra $x$.
\end{definition}
The following result  justifies our axiomatization in~\cref{def:expSitX}: adequacy, a property that is a prerequisite in most usage scenarios of modal logics, follows easily from the axiomatization itself.
\begin{propositionrep}
	\label{prop:adequacy}
Any expressivity situation $\esit$ in \cref{def:expSitX} is adequate.
 \myqed
\end{propositionrep}
\begin{proof}
	Fix a $B$-coalgebra $x\colon X\to BX$.
	It suffices to show that, for any $\varphi\in L_{\esit}$, $\nu(x^*\co\codlif{B}{\bOmega}{\tau})\sqsubseteq\sem{\varphi}_{x}^{*}\bOmega$ holds.
	We show this by structural induction on $\varphi$.
	
	Assume $\varphi=\sigma(\varphi_{1},\dotsc,\varphi_{\rank(\sigma)})$ where $\sigma\in\Sigma$.
	By IH, for each $i=1,\dots,\rank(\sigma)$, $\nu(x^*\co\codlif{B}{\bOmega}{\tau})\sqsubseteq\sem{\varphi_{i}}_{x}^{*}\bOmega$ holds, and thus there exists an arrow $h_{i}\colon\nu(x^*\co\codlif{B}{\bOmega}{\tau})\to\bOmega$ in $\Cat{E}$ that satisfies $ph_{i}=\sem{\varphi_{i}}_{x}$.
	Take an arrow $g_\sigma\colon\bOmega^{\rank(\sigma)}\to\bOmega$ satisfying $pg_\sigma=f_\sigma$ (its existence is required in \cref{def:expSitX}).
	Consider the arrow $g_\sigma\co\langle h_1,\dots,h_{\rank(\sigma)}\rangle\colon\nu(x^*\co\codlif{B}{\bOmega}{\tau})\to\bOmega$.
	Since $p$ sends this arrow to $f_\sigma\co\langle \sem{\varphi_1}_x,\dots,\sem{\varphi_{\rank(\sigma)}}\rangle=\sem{\varphi}_x$, $\nu(x^*\co\codlif{B}{\bOmega}{\tau})\sqsubseteq\sem{\varphi}_{x}^{*}\bOmega$ holds.
	
	Assume $\varphi=\heartsuit_{\lambda}\varphi'$ where $\lambda\in\Lambda$.
	By IH, $\nu(x^*\co\codlif{B}{\bOmega}{\tau})\sqsubseteq\sem{\varphi'}_{x}^{*}\bOmega$ holds.
	Applying $\codlif{B}{\bOmega}{\tau}$ to both sides yields
	\begin{align*}
		\codlif{B}{\bOmega}{\tau}\nu(x^*\co\codlif{B}{\bOmega}{\tau})&\sqsubseteq\codlif{B}{\bOmega}{\tau}\sem{\varphi'}_{x}^{*}\bOmega \\
		&=\bigsqcap_{\lambda'\in\Lambda,h\colon\sem{\varphi'}_{x}^{*}\bOmega\to\bOmega}(B(ph))^*\tau_{\lambda'}^*\bOmega\\
		&\sqsubseteq(B\sem{\varphi'}_{x})^*\tau_{\lambda}^*\bOmega.\\
	\end{align*}
	Then by applying $x^*$ to both sides we obtain the claim: \[
		\nu(x^*\co\codlif{B}{\bOmega}{\tau}) = x^*\codlif{B}{\bOmega}{\tau}\nu(x^*\co\codlif{B}{\bOmega}{\tau}) \sqsubseteq x^*(B\sem{\varphi'}_{x})^*\tau_{\lambda}^*\bOmega = \sem{\varphi}_x^*\bOmega.
	\]
	This concludes the induction.
\end{proof}

\begin{example}
	\label{ex:expressiveKMM}
	Recall \cref{ex:codensityBisimKMM}. Expressivity of the expressivity situation $\esit_{\mathrm{KMM}}$ means that,  for each $x\colon X\to BX$ and each pair $(s,t)\in X^2$ of states, the inequality $d_x(s,t)\le d^L_x(s,t)$ holds  between the behavioral and logical distances (``$d^L_x$ is more discriminating''). Adequacy means that, for each $x$ and $(s,t)$, $d_x(s,t)\ge d^L_x(s,t)$ holds.
\end{example}

\section{Approximation in Quantitative Expressivity}
\label{sec:approxAndExp}
In this section, based on the axiomatization in~\cref{sec:expSit}, we present a fibrational notion of \emph{approximating family of observations}. The notion axiomatizes and unifies the  ``approximation'' properties that are key steps in many recent quantitative expressivity proofs, such as in~\cite{KonigM18,ClercFKP19,WildSP018,WildSP019}.

We then proceed to present two proof principles for expressivity---\emph{Knaster--Tarski} and \emph{Kleene}---that mirror two classic characterizations of greatest fixed points. These proof principles make a large part of an expressivity proof routine. The remaining technical challenges are 1) choosing a suitable propositional signature and 2) identifying suitable approximating families; our general framework singles out these technical challenges and thus eases the efforts for addressing them.

We recall the two characterizations of gfps.
\begin{theorem}\label{thm:KnasterTarskiKleeneX}
	Let $(L,\sqsubseteq)$ be a complete lattice, and $f\colon L\to L$ be a monotone function. 
	\begin{enumerate}
		\item \textbf{(Knaster--Tarski)} The set $\{l\in L\mid l\sqsubseteq f(l)\}$ of post-fixed points is a complete lattice. Its maximum $z$ satisfies $z=f(z)$, hence $z$ is the greatest fixed point $\nu f$. 
		\item \textbf{(Kleene)} Consider the following $\omega^{\op}$-chain in $L$. 
		\begin{equation}\label{eq:kleeneSeqX}
			\top \;\sqsupseteq\; f(\top)
			\;\sqsupseteq\; f^{2}(\top)
			\;\sqsupseteq\; \cdots
		\end{equation}
		If $f$ preserves the meet $\bigsqcap_{i\in \omega^{\op}}f^{i}(\top)$, then $\bigsqcap_{i\in \omega^{\op}}f^{i}(\top)$ is the greatest fixed point $\nu f$. 
	\end{enumerate}
\end{theorem}

\subsection{Approximating Family of Observations}\label{subsec:approximableFamily}

Our categorical notion of \emph{approximating family of observations} designates a ``good'' subset $S\subseteq\Cat{C}(X,\Omega)$ of $\Omega$-valued observations in a suitable sense. We will be asking if the set $\{\sem{\varphi}_{x}\colon X\to \Omega\mid\varphi\in L'\}$ of ``logical observations'' is approximating or not, where $L'$ is some set of modal formulas.

\begin{definition}[approximating family]
	\label{def:approximable}
	Let $\esit$ be an expressivity situation in \cref{def:expSitX} and $X$ be an object of $\Cat{C}$.
	A subset $S\subseteq\Cat{C}(X,\Omega)$ is an \emph{approximating family of observations}, or simply \emph{approximating}, if, for every morphism 
\begin{equation}\label{eq:approximable}
 	h\colon\;\textstyle \bigl(\,\bigsqcap_{k\in S}k^*\bOmega\,\bigr) \longrightarrow \bOmega \end{equation}	
of $\Cat{E}$ and every $\lambda\in\Lambda$, the following inequality holds: 
\begin{equation}\label{eq:approximable2}
 	\textstyle\bigsqcap_{k'\in S, \lambda'\in\Lambda} (\tau_{\lambda'}\co Bk')^*\bOmega\;\sqsubseteq\; (\tau_\lambda\co B(ph))^*\bOmega.
\end{equation}
Note that $k\colon X\to \Omega$ is a $\Cat{C}$-arrow while $h$ is an $\Cat{E}$-arrow.
\end{definition}

Some explanation is in order.
Intuitively, in the definition above, the set $S$ is a set of ``logical'' observations.
Each $h$ as in~\cref{eq:approximable} is a ``non-logical'' legitimate observation.
For such $h$, the r.h.s.~of~\cref{eq:approximable2} is the information obtained from the observation $h$.
(Note the way $h$ is used: it is not $h^*\bOmega$, but $(\tau_\lambda\co B(ph))^*\bOmega$.
This corresponds to~\cref{eq:codensity}.)
On the other hand, the l.h.s.~of~\cref{eq:approximable2} is the information from ``logical'' observations.
Thus, an intuitive meaning of the definition above is that no ``non-logical'' observation gives any additional information.
In many cases, the ``logical'' observations in $S$ approximate each ``non-logical'' ones $h$. See \cref{rem:twoStepsInApprox} for details.

Another intuition is given in terms of the codensity bisimilarity game  (\cref{fact:codensityGame}). Roughly, $S$ being an approximating family says that Spoiler may restrict its moves to $S\subseteq\Cat{C}(X,\Omega)$.

\begin{remark}\label{rem:twoStepsInApprox}
 In many examples, $S$ being an approximating family is proved in the following two steps: 1) showing that $ph$ can be approximated by observations in $S$; and 2) this approximation is preserved along the lifting  $k\mapsto\tau_{\lambda}\co Bk$ of observations over $X$ to those over $BX$. The former step is usually the harder one, and proved via arguments specific to the current situation (pseudometric spaces, measurable spaces, etc.). 
\end{remark}

\begin{example}
	\label{ex:approximableKMM}
	Recall \cref{ex:expressiveKMM}.
	Let $X$ be a set and $S\subseteq\Set(X,[0,1])$.
	In this case, $\bigsqcap_{k\in S}k^*\bOmega\in(\PMetb)_X$ is the pseudometric $d_S$ given by $d_S(x,y)=\sup_{k\in S}d_e(k(x),k(y))$.
	Therefore, in order to show  $S$ being an approximating family, we have to recover the pseudometric induced by $h\colon(X,d_S)\to([0,1],d_e)$ from observations in $S$, for each $h$.
	
	In \cref{prop:exampleKoenigMMApproximability} later, it will turn out that $S$ is approximating if the following hold (under \cref{asm:exampleKoenigMM}):
	\begin{itemize}
		\item $S$ is closed under  the four operations
		$\top$, $\min$, $\neg$, and $\ominus q$ for every
		$q \in \mathbb{Q} \cap [0,1]$.
		\item $(X,d_S)$ is totally bounded.
\end{itemize}
	In this case, any $h\colon(X,d_S)\to([0,1],d_e)$ can be uniformly approximated by a countable sequence of arrows in $S$. Moreover, this approximation is preserved by the lifting $k\mapsto\tau_{\lambda}\co Bk$ (this is what we require in \cref{asm:exampleKoenigMM}). These two facts establish that $S$ is approximating (cf.~\cref{rem:twoStepsInApprox}). See  \cref{prop:exampleKoenigMMApproximability} for details.
\end{example}

\subsection{The Knaster--Tarski Proof Principle for Expressivity}
From the Knaster--Tarski theorem, we can derive the following simple expressivity proof principle. Its proof is by showing that the logical equivalence $\LEq{\esit}{x}$ is a suitable invariant and thus underapproximates the codensity bisimilarity (\cref{thm:KnasterTarskiKleeneX}).

\begin{theoremrep}[the Knaster--Tarski proof principle]
	\label{thm:knasterTarskiApproxPrinciple}
	Let $\esit$ be an expressivity situation in \cref{def:expSitX} and $x\colon X\to BX$ be a $B$-coalgebra. If $\{\sem{\varphi}_x~|~\varphi\in L_{\esit}\}\subseteq\Cat{C}(X,\Omega)$ is approximating, 
	then $\esit$ is expressive for $x$. 
\myqed
\end{theoremrep}
\begin{proof}
	We show $\nu(x^*\co\codlif{B}{\bOmega}{\tau})\sqsupseteq\bigsqcap_{\varphi\in L_{\esit}}\sem{\varphi}_{x}^{*}\bOmega$.
	By the Knaster--Tarski theorem (\cref{thm:KnasterTarskiKleeneX}), it suffices to show \[
		x^*\codlif{B}{\bOmega}{\tau}\left(\bigsqcap_{\varphi\in L_{\esit}}\sem{\varphi}_{x}^{*}\bOmega\right)\sqsupseteq\bigsqcap_{\varphi\in L_{\esit}}\sem{\varphi}_{x}^{*}\bOmega.
	\]
	Since the l.h.s.~is equal to \[
		\bigsqcap_{\lambda\in\Lambda,h\colon\bigsqcap_{\varphi\in L_{\esit}}\sem{\varphi}_{x}^{*}\bOmega\to\bOmega}x^*(B(ph))^*\tau_{\lambda}^*\bOmega,
	\] it suffices to show \[
		x^*(B(ph))^*\tau_{\lambda}^*\bOmega \sqsupseteq \bigsqcap_{\varphi\in L_{\esit}}\sem{\varphi}_{x}^{*}\bOmega.
	\] for each $\lambda\in\Lambda$ and $h\colon\bigsqcap_{\varphi\in L_{\esit}}\sem{\varphi}_{x}^{*}\bOmega\to\bOmega$.
	
	The set $\{\sem{\varphi}_{x}~|~\varphi\in L_{\esit}\}\subseteq\Cat{C}(X,\Omega)$ being approximating implies the following lower bound of the l.h.s.:
	\begin{align*}
		x^*(B(ph))^*\tau_{\lambda}^*\bOmega &\sqsupseteq \bigsqcap_{\varphi'\in L_{\esit}, \lambda'\in\Lambda} x^*(\tau_{\lambda'}\co B\sem{\varphi'}_x)^*\bOmega \\
		&= \bigsqcap_{\varphi'\in L_{\esit}, \lambda'\in\Lambda} \sem{\heartsuit_{\lambda'}\varphi'}_x^*\bOmega \\
		&\sqsupseteq \bigsqcap_{\varphi\in L_{\esit}}\sem{\varphi}_{x}^{*}\bOmega. \\
	\end{align*}
	This concludes the proof.
\end{proof}
\noindent The theorem's applicability  hinges on whether we can show that the set $\{\sem{\varphi}_{x}~|~\varphi\in L_{\esit}\}\subseteq\Cat{C}(X,\Omega)$ is an approximating family (where $\varphi$ ranges over all formulas).
We use the theorem for the examples in~\S{}\ref{sec:exMarkovProcesses} \&~\ref{sec:exBisimUnif}.

\subsection{The Kleene Proof  Principle for Expressivity} \label{sec:kleenepp}
To make use of Kleene theorem, we have to consider 
\begin{equation}
	\label{eq:kleeneFinalSeq}
	\top \;\sqsupseteq\; (x^*\co\codlif{B}{\bOmega}{\tau})(\top)
	\;\sqsupseteq\; (x^*\co\codlif{B}{\bOmega}{\tau})^2(\top)
	\;\sqsupseteq\; \cdots
\end{equation}
where the functor $x^*\co\codlif{B}{\bOmega}{\tau}$ is from~\cref{def:codensityBisim}. 
We also have to assume that this sequence \emph{stabilizes after $\omega$ steps}, i.e., $\bigsqcap_{i<\omega}(x^*\co\codlif{B}{\bOmega}{\tau})^i(\top)$ is a fixed point of $x^*\co\codlif{B}{\bOmega}{\tau}$.

 We  stratify $L_{\esit}$  corresponding to the sequence~\cref{eq:kleeneFinalSeq}.

\begin{definition}[depth]
	Let $\esit$ be an expressivity situation in \cref{def:expSitX}.
	For each $\varphi\in L_{\esit}$, the \emph{depth of $\varphi$} $\depth(\varphi)$ is a natural number defined inductively as follows:
	\begin{align*}
		&\depth(\sigma(\varphi_{1},\dotsc,\varphi_{\rank(\sigma)})) \\
		&\quad=\max(\depth(\varphi_{1}),\dotsc,\depth(\varphi_{\rank(\sigma)}))  &(\sigma\in\Sigma) \\
		&\depth(\heartsuit_{\lambda}\varphi)=\depth(\varphi)+1&(\lambda\in\Lambda) 
	\end{align*}
	For $\sigma\in\Sigma$ with $\rank(\sigma)=0$, $\depth(\sigma())$ is defined to be $0$.
\end{definition}

We formulate the following proof principle. Unlike  Knaster--Tarski (\cref{thm:KnasterTarskiKleeneX}), it uses an explicit induction on the depth $i$. Its proof is therefore more involved but not much more.

\begin{theoremrep}[the Kleene proof principle]\label{thm:kleeneApproxPrinciple}
	Let $\esit$ be an expressivity situation as in \cref{def:expSitX} and $x\colon X\to BX$ be a $B$-coalgebra.
	Assume that the chain~\cref{eq:kleeneFinalSeq} in $\Cat{E}_X$ 
stabilizes after $\omega$ steps.
If the set $\{\sem{\varphi}_{x}~|~\varphi\in L_{\esit},\depth(\varphi)\le i\}\subseteq\Cat{C}(X,\Omega)$ is approximating  for each $i$, 	then $\esit$ is expressive for $x$. \myqed
\end{theoremrep}
\begin{proof}
	By Kleene theorem (\cref{thm:KnasterTarskiKleeneX}), it suffices to show $(x^*\co\codlif{B}{\bOmega}{\tau})^i(\top)\sqsupseteq\bigsqcap_{\varphi\in L_{\esit}}\sem{\varphi}_{x}^{*}\bOmega$ for each $i$.
	We show
	\begin{equation}\label{eq:kleeneApproxPrincipleProof1}
		 (x^*\co\codlif{B}{\bOmega}{\tau})^i(\top)\sqsupseteq\bigsqcap_{\varphi\in L_{\esit}, \depth(\varphi)\le i}\sem{\varphi}_{x}^{*}\bOmega 
	\end{equation}
	by induction on \(i\).
	
	For $i=0$,~\cref{eq:kleeneApproxPrincipleProof1} is trivial.
	
	Assume that~\cref{eq:kleeneApproxPrincipleProof1} holds for $i=j$ and we show it also holds for $i=j+1$.
	Start with~\cref{eq:kleeneApproxPrincipleProof1} for $i=j$.
	Applying $x^*\co\codlif{B}{\bOmega}{\tau}$ to both sides of it we obtain \[
		(x^*\co\codlif{B}{\bOmega}{\tau})^{j+1}(\top)\sqsupseteq(x^*\co\codlif{B}{\bOmega}{\tau})\bigsqcap_{\varphi\in L_{\esit}, \depth(\varphi)\le j}\sem{\varphi}_{x}^{*}\bOmega.
	\]
	Here, expanding the definition of the r.h.s.~we get \[
		(x^*\co\codlif{B}{\bOmega}{\tau})\bigsqcap_{\varphi\in L_{\esit}, \depth(\varphi)\le j}\sem{\varphi}_{x}^{*}\bOmega = \bigsqcap_{\lambda\in\Lambda,h\colon\bigsqcap_{\varphi\in L_{\esit}, \depth(\varphi)\le j}\sem{\varphi}_{x}^{*}\bOmega\to\bOmega}x^*(B(ph))^*\tau_{\lambda}^*\bOmega.
	\]
	Now let $\lambda\in\Lambda$ and $h\colon\bigsqcap_{\varphi\in L_{\esit}, \depth(\varphi)\le j}\sem{\varphi}_{x}^{*}\bOmega\to\bOmega$.
	That $\{\sem{\varphi}_{x}~|~\varphi\in L_{\esit},\depth(\varphi)\le j\}\subseteq\Cat{C}(X,\Omega)$ is an approximating family yields \begin{align*}
		x^*(B(ph))^*\tau_{\lambda}^*\bOmega &\sqsupseteq \bigsqcap_{\lambda'\in\Lambda,\varphi\in L_{\esit}, \depth(\varphi)\le j}x^*(B\sem{\varphi}_{x})^*\tau_{\lambda'}^*\bOmega \\
		&= \bigsqcap_{\lambda'\in\Lambda,\varphi\in L_{\esit}, \depth(\varphi)\le j}\sem{\heartsuit_{\lambda'}\varphi}_{x}^*\bOmega \\
		&\sqsupseteq \bigsqcap_{\varphi'\in L_{\esit}, \depth(\varphi')\le j+1}\sem{\varphi'}_{x}^*\bOmega.
	\end{align*}
	Thus we have
	\begin{align*}
		(x^*\co\codlif{B}{\bOmega}{\tau})^{j+1}(\top)&\sqsupseteq(x^*\co\codlif{B}{\bOmega}{\tau})\bigsqcap_{\varphi\in L_{\esit}, \depth(\varphi)\le j}\sem{\varphi}_{x}^{*}\bOmega \\
		&= \bigsqcap_{\lambda\in\Lambda,h\colon\bigsqcap_{\varphi\in L_{\esit}, \depth(\varphi)\le j}\sem{\varphi}_{x}^{*}\bOmega\to\bOmega}x^*(B(ph))^*\tau_{\lambda}^*\bOmega \\
		&\sqsupseteq \bigsqcap_{\varphi'\in L_{\esit}, \depth(\varphi')\le j+1}\sem{\varphi'}_{x}^*\bOmega.
	\end{align*}
	This concludes the induction.
\end{proof}
\noindent
 In~\cref{thm:kleeneApproxPrinciple}, we require that $\{\sem{\varphi}_{x}~|~\varphi\in L_{\esit},\depth(\varphi)\le i\}$ is approximating for each depth $i$; this is often easier than the case where $\varphi$ ranges over all formulas (as in \cref{thm:knasterTarskiApproxPrinciple}). We use the theorem for the example in~\cref{sec:exampleKoenigMM}.

\begin{example}
Sufficient conditions for being an approximating family were given in~\cref{ex:approximableKMM}. Combined with \cref{thm:kleeneApproxPrinciple}, it yields expressivity (\cref{cor:exampleKoenigMMExpressive}), one of the main results of~\cite{KonigM18}.

\end{example}

\begin{remark}
	\label{rem:chainLength}
 In \cref{thm:kleeneApproxPrinciple}, we assumed the stabilization of the chain~\cref{eq:kleeneFinalSeq} at length $\omega$. This assumption turns out to be benign, essentially because our modal formulas all have a finite depth (\cref{subsec:syntaxSemOfLogic}).
Specifically we can show the following: if the logic  $L_{\esit}$ is expressive for $x\colon X\to BX$, then the chain~\cref{eq:kleeneFinalSeq}  stabilizes after $\omega$ steps. 
See Appendix~\ref{appendix:remchainLength}.

\end{remark}

\section{Expressivity for the Kantorovich Bisimulation Metrics}\label{sec:exampleKoenigMM}

This section shows how one of the main results in~\cite{KonigM18}, expressivity of a real-valued logic w.r.t.~bisimulation metric, is proved by our Kleene proof principle (\cref{thm:kleeneApproxPrinciple}).
See also \cref{ex:syntaxKMM,ex:formulaSemanticsKMM,ex:logicalEqKMM}.

\begin{definition}
	\label{def:exampleKoenigMMSetting}
	Define an expressivity situation $\esit_{\mathrm{KMM}}$ by:
	\begin{itemize}
		\item Its  fibration is $\PMetb\to\Set$ (\cref{ex:fibr}).
		\item Its truth-value object is $[0,1]$ and its observation predicate is $d_e$, the usual Euclidean metric on $[0,1]$.
		\item The ranked alphabet of its propositional connectives is $\Sigma=\{\top^0,\min^2,\neg^1\}\cup\{(\ominus q)^1~|~q\in\mathbb Q\cap[0,1]\}$. Its propositional structure $(f_\sigma\colon[0,1]^{\rank(\sigma)}\to[0,1])_{\sigma\in\Sigma}$ is specified by:
		\begin{align*}
			f_{\top}()&=1&f_{\min}(x,y)&=\min(x,y)\\
			f_{\neg}(x)&=1-x&f_{\ominus q}(x)&=\max(x-q,0)
		\end{align*}
		\item The behavior functor $B\colon\Set\to\Set$, the set of its modality indices $\Lambda$, and its observation modalities $(\tau_\lambda\colon B[0,1]\to[0,1])_{\lambda\in\Lambda}$ are arbitrary.
	\end{itemize}
\end{definition}

The modal logic $L_{\esit_{\mathrm{KMM}}}$ is the same as the logic $\mathcal{M}(\Lambda)$ in~\cite[Table 1]{KonigM18}.
What they call an \emph{evaluation map} $\gamma\in\Gamma$ corresponds to an observation modality $\tau_\lambda(\lambda\in\Lambda)$ in our framework.
Thus, the fibrational logical equivalence $\LEq{\esit_{\mathrm{KMM}}}{\alpha}$ (\cref{def:logicalEq}) coincides with the \emph{logical distance} $d_\alpha^L$~\cite[Def.~25]{KonigM18} for a coalgebra $\alpha\colon X\to BX$.

Moreover, the codensity lifting $\codlif{B}{d_e}{\tau}$ specializes to the {\em Kantorovich lifting} by \cite{BaldanBKK18}:
\begin{align*}
	\codlif{B}{d_e}{\tau}(X,d) &= (BX, d_B)\quad\text{where}\\
	d_B(t_1,t_2) &=
	\textstyle\sup_{\lambda,h}
	d_e(\tau_\lambda ((B h) (t_1)),\tau_\lambda((Bh) (t_2))).
\end{align*}
In the above sup, $\lambda,h$ ranges over $\Lambda$ and
$\PMetb((X,d),([0,1],d_e))$, respectively.
Thus, the codensity bisimilarity $\CBisim{d_e}{\tau}{\alpha}$ (\cref{def:codensityBisim}) recovers the definition of the \emph{behavioral distance} $d_\alpha$~\cite[Def.~22]{KonigM18} for a coalgebra $\alpha\colon X\to BX$.

From \cref{prop:adequacy} we obtain:
\begin{corollary}
	For $\alpha\colon X\to BX$, $d_\alpha \ge d_\alpha^L$ holds.\myqed
\end{corollary}

As mentioned in~\cite{KonigM18}, it is harder to prove expressivity.
We use the Kleene proof principle (\cref{thm:kleeneApproxPrinciple}) here.
For this argument to work, we have to make further assumptions.
\begin{assumption}
	\label{asm:exampleKoenigMM}
	For $\esit_{\mathrm{KMM}}$, assume the following:
	\begin{enumerate}
		\item $\Lambda$ is finite.
		\item\label{item:exampleKoenigMM} If a sequence $k_i$ of functions of type $X\arrow[0,1]$ uniformly converges into $l$,
		then $\tau_\lambda\co Bk_i\colon BX\arrow[0,1]$ uniformly converges into $\tau_\lambda\co Bl$ for each $\lambda\in\Lambda$.
	\end{enumerate}
\end{assumption}

In particular, condition \ref{item:exampleKoenigMM} above is satisfied if each $\tau_\lambda$ \emph{induces a
	non-expansive predicate lifting}~\cite[Def.~17]{KonigM18}.

The notion of \emph{total boundedness} below is pivotal in~\cite{KonigM18}.
\begin{definition}[{from~\cite[Def.~28]{KonigM18}}]
	$(X,d)\in\PMetb$ is \emph{totally bounded} if, for any
	$\varepsilon>0$, there is a finite set $F_\varepsilon\subseteq X$
	satisfying the following: for each $x\in X$, there is
	$y\in F_\varepsilon$ such that $d(x,y)<\varepsilon$.
\end{definition}

A critical step in their proof used a Stone--Weierstrass-like property of
totally bounded spaces.

\begin{propositionrep}\label{prop:exampleKoenigMMDensity}
	Let $(X,d)$ be a totally bounded pseudometric space. A subset $S \subseteq \PMetb((X,d), ([0,1],d_e))$ is dense in the topology of uniform convergence if the following are satisfied:
	\begin{enumerate}
		\item $S$ is closed under the four operations
		$\top$, $\min$, $\neg$ and $\ominus q$ for every
		$q \in \mathbb{Q} \cap [0,1]$;
		\item for every $h \in\PMetb((X,d), ([0,1],d_e))$, and every $x,y \in X$, we
		have
		\[
		d_e(h(x),h(y)) \leq \textstyle\sup_{g \in S} d_e(g(x),g(y)) \,.\mymyqed
		\]
	\end{enumerate}
\end{propositionrep}
\begin{proof}
	By~\cite[Lemma 5.8]{WildSP018}, it suffices to show that, for each $h\in\PMetb((X,d), ([0,1],d_e))$, each $\delta>0$, and each pair of points $x,y\in X$, there is $g\in S$ such that $d_e(h(x),g(x))\le\delta$ and $d_e(h(y),g(y))\le\delta$ hold.
	
	Without loss of generality, we can assume $h(x)\ge h(y)$.
	Let $\gamma=h(x)-h(y)$. Since $\gamma\ge 0$, $\gamma=d_e(h(x),h(y))$.
	By the second assumption, there is $f$ such that $\gamma-\delta\le d_e(f(x),f(y))$.
	Since $S$ is closed under $\neg$, we can assume that $f(x)\ge f(y)$ without loss of generality.
	This implies $\gamma-\delta\le f(x)-f(y)$.
	
	Now, we do a case analysis.
	
	Firstly, assume $f(y)\ge h(y)$.
	Take $r,s\in\mathbb{Q}\cap[0,1]$ satisfying $f(y)-h(y)-\delta\le r\le f(y)-h(y)$ and $h(x)\le s\le h(x)+\delta$.
	Then $g=\min(f\ominus r,s)$ is what we want.
	
	Secondly, assume $f(y)<h(y)$.
	Take $r,s\in\mathbb{Q}\cap[0,1]$ satisfying $h(y)-f(y)-\delta\le r\le h(y)-f(y)$ and $h(x)\le s\le h(x)+\delta$.
	Then $g=\min(\neg((\neg f)\ominus r),s)$ is what we want.
\end{proof}

In our framework, this can be stated in the following form:
\begin{propositionrep}
	\label{prop:exampleKoenigMMApproximability}
	Assume the setting of \cref{def:exampleKoenigMMSetting}.
	Let $X\in\Set$.
	Under \cref{asm:exampleKoenigMM}, a subset $S\subseteq\Set(X,[0,1])$ is approximating if the following hold:
	\begin{itemize}
		\item $S$ is closed under  the four operations
		$\top$, $\min$, $\neg$ and $\ominus q$ for every
		$q \in \mathbb{Q} \cap [0,1]$.
		\item $(X,d_S)$ is totally bounded, where $d_S(x,y)=\sup_{k\in S}d_e(k(x),k(y))$.\myqed
	\end{itemize}
\end{propositionrep}
\begin{proof}
	Fix $h\colon (X,d_S) \to ([0,1], d_e)$, $\lambda\in\Lambda$, and $(z,w)\in (BX)^2$.
	It suffices to show the following:
	\begin{equation}
		\label{eq:exampleKoenigApproximability}
		\sup_{k\in S,\lambda'\in\Lambda}d_e(\tau_{\lambda'}((Bk)(z)),\tau_{\lambda'}((Bk)(w))) \ge d_e(\tau_{\lambda}((Bh)(z)),\tau_{\lambda}((Bh)(w))).
	\end{equation}

	Use \cref{prop:exampleKoenigMMDensity} for $d=d_S$.
	This ensures the existence of a sequence $(k_n\colon X\to[0,1])_{n=1,2,\dots}$ that uniformly converges
	to $ph$ as $n\to\infty$.
	
	Fix $\varepsilon>0$.
	By \cref{asm:exampleKoenigMM}, the sequence $(\tau_\lambda\co k_n)_{n=1,2,\dots}$ also uniformly converges to $\tau_\lambda\co (ph)$ as $n\to\infty$.
	Thus, we can fix $n$ so that $d_e(\tau_{\lambda}((Bk)(z)),\tau_{\lambda}((Bk_n)(z)))<\varepsilon$ and $d_e(\tau_{\lambda}((Bh)(w)),\tau_{\lambda}((Bk_n)(w)))<\varepsilon$ both hold.
	From the triangle inequality, we obtain $d_e(\tau_{\lambda}((Bk_n)(z)),\tau_{\lambda}((Bk_n)(w))) \ge d_e(\tau_{\lambda}((Bh)(z)),\tau_{\lambda}((Bh)(w))) + 2\varepsilon$.
	
	Since $\varepsilon$ is arbitrary, we have \cref{eq:exampleKoenigApproximability}.
\end{proof}

From now we use some facts on totally bounded space.
Using the variation of Arzel\`{a}--Ascoli theorem~\cite[Lemma 5.6]{WildSP018}
for totally bounded spaces, we can show the following:
\begin{fact}
	\label{fact:exampleKoenigMM}
	Under \cref{asm:exampleKoenigMM}, if $(X,d)\in\PMetb$ is totally bounded, then
	\begin{itemize}
		\item $\codlif{B}{d_e}{\tau}(X,d)$ is also totally bounded. \footnote{
			Here the finiteness of the number of modalities is crucial.
			When $\Lambda$ is infinite,  the Kantorovich lifting
			does not preserve total boundedness.  See~\cref{subsubsec:totalBoundedness}.
		}
		\item If $(X,d)\sqsubseteq(X,d')$, $(X,d')$ is also totally bounded.
	\end{itemize}
\end{fact}

These enable us to use \cref{thm:kleeneApproxPrinciple}:

\begin{propositionrep}
	\label{prop:exampleKoenigMMStepApproximable}
	Let $x\colon X\to BX$ be a coalgebra.
	Under \cref{asm:exampleKoenigMM}, for each $i$, $\{\sem{\varphi}~|~\varphi\in L_{\esit_{\mathrm{KMM}}},\depth(\varphi)\le i\}\subseteq\Set(X,[0,1])$ is approximating.\myqed
\end{propositionrep}
\begin{proof}
	By induction, for each $i$, $(\codlif{B}{d_e}{\tau})^i(\top)\in(\PMetb)_X$ is totally bounded.
	By the stepwise adequacy (\cref{rem:chainLength} \& \cref{appendix:remchainLength}) and \cref{fact:exampleKoenigMM}, for each $i$, $\bigsqcap_{\varphi\in L_{\esit_{\mathrm{KMM}}},\rank(\varphi)\le i}\sem{\varphi}_x^*d_e$ is also totally bounded.
	From this and \cref{prop:exampleKoenigMMApproximability}, we can show that the desired set is  approximating.
\end{proof}

\begin{corollaryrep}[{from~\cite[Thm.~32]{KonigM18}}]
	\label{cor:exampleKoenigMMExpressive}
	Let $\alpha\colon X\to BX$ be a coalgebra.
	Assume that the sequence
	$\top\sqsupseteq (x^{*}\codlif{B}{d_e}{\tau})(\top) \sqsupseteq
	(x^{*}\codlif{B}{d_e}{\tau})^{2}(\top)\sqsupseteq\cdots$ stabilizes after $\omega$
	steps (as in Thm.~\ref{thm:kleeneApproxPrinciple}).
	Then, under \cref{asm:exampleKoenigMM}, $d_\alpha \le d_\alpha^L$ holds.
	In particular, $d_\alpha$ is characterized as the greatest pseudometric that makes all $\sem{\varphi}_\alpha$ nonexpansive.\myqed
\end{corollaryrep}
\begin{proof}
	Use \cref{thm:kleeneApproxPrinciple}.
	The premises are satisfied by \cref{prop:exampleKoenigMMStepApproximable}.
\end{proof}

\section{Expressivity for Markov Process Bisimilarity}\label{sec:exMarkovProcesses}

This section shows how one of the main results in~\cite{ClercFKP19}, expressivity of probabilistic modal logic w.r.t.~bisimilarity of labelled Markov process, is proved by our Knaster--Tarski proof principle (\cref{thm:knasterTarskiApproxPrinciple}).

Throughout this section, fix a set $A$ of labels.

\begin{definition}
	\label{def:exMarkovProcessesSetting}
	Define an expressivity situation $\esit_{\mathrm{CFKP}}$ by:
	\begin{itemize}
		\item Its fibration is $\EqRelMeas\arrow\Meas$ (\cref{ex:fibr}).
		\item Its behavior functor $B\colon\Meas\to\Meas$ is $BX=(\SGiry X)^A$, where $\SGiry$ is the variation of
		Giry functor, which sends each measurable space to its space of
		subdistributions.
		\item Its truth-value object is $2=\{0,1\}$ with all subsets measurable and its observation predicate is the equality relation $\Eq_2$ on $2$.
		\item The ranked alphabet of its propositional connectives is $\Sigma=\{\top^0,\wedge^2\}$. Its propositional structure $(f_{\top},f_{\wedge})$ is specified as the usual boolean operations.
		\item The set of its modality indices is $A\times (\mathbb{Q}\cap[0,1])$. For each $(a,r)\in A\times (\mathbb{Q}\cap[0,1])$, the observation modality $\tau_{a,r}\colon(\SGiry 2)^A\to 2$ is defined by
		\[ \tau_{a,r}((\mu_a)_{a\in A}) =
		\mathrm{thr}_r(\mu_a(\{1\})),
		\] where $\mathrm{thr}_r(s)=1$ if and only if $s>r$.
	\end{itemize}
\end{definition}

Note that a \emph{labelled Markov process (LMP) with label set $A$}~\cite[Definition 5.1]{ClercFKP19} is the same as
$B$-coalgebra.
The modal logic $L_{\esit_{\mathrm{CFKP}}}$ (\cref{def:syntax}) has the following syntax: \[
	\varphi_{1},  \varphi_{2}\;::=\;\top\mid\wedge(\varphi_{1},\varphi_{2}) \mid\;\heartsuit_{a,r}\varphi_{1}\;((a,r)\in A\times (\mathbb{Q}\cap[0,1]))
\]
So if we identify $\heartsuit_{a,r}$ with $\langle a \rangle_r$ in the original notation, this recovers the syntax of $\mathrm{PML}_\wedge$ defined in~\cite[Def.~2.3]{ClercFKP19}.
Under this identification, the semantics (\cref{def:formulaSemantics}) is also essentially the same as the original logic: $\sem{\varphi}_x(s)=1\iff s\vDash\varphi$ holds for any LMP $x\colon X\to(\SGiry X)^A$, any point $s\in X$, and any formula $\varphi$.
The fibrational logical equivalence (\cref{def:logicalEq}) can be concretely represented as \[
	\LEq{\esit_{\mathrm{CFKP}}}{x} = \{(s,t)~|~\forall\varphi\in L_{\esit_{\mathrm{CFKP}}},s\vDash\varphi\iff t\vDash\varphi\}.
\]

By expanding the definition of the codensity lifting (\cref{def:codensityLifting}) of $(\SGiry \place)^A$, we can see that it coincides with the one used to define \emph{probabilistic bisimulation}:
\begin{proposition}
	\label{prop:codensityisprobbisim}
	The codensity lifting $\codlif{(\SGiry \place)^A}{\Eq_2}{\tau}$ satisfies the
	following: for each
	$(\mu_a)_{ a\in A},(\nu_ a)_{ a\in A}\in
	(\SGiry X)^A$, they are equivalent in $\codlif{(\SGiry \place)^A}{\Eq_2}{\tau}(X,R)$ if and
	only if, for each $ a\in A$ and each $R$-closed measurable
	set $S\subseteq X$, $\mu_ a(S)=\nu_ a(S)$ holds.\myqed
\end{proposition}
Thus the codensity bisimilarity $\CBisim{\Eq_2}{\tau}{x}$ (\cref{def:codensityBisim}) coincides with the probabilistic bisimilarity used in~\cite{ClercFKP19}.

From \cref{prop:adequacy}, we readily obtain the following:
\begin{corollary}
	Let $x\colon X\to(\SGiry X)^A$ be an LMP.
	If $s,t\in X$ are probabilistically bisimilar, for any $\varphi\in L_{\esit_{\mathrm{CFKP}}}$, $s\vDash\varphi\iff t\vDash\varphi$ holds.\myqed
\end{corollary}

To show expressivity, we first have to review some mathematical key facts.
In the rest of this section, we write $\sigma(\mathscr{E})$ for the
$\sigma$-algebra generated by a family of sets $\mathscr{E}$.

\begin{definition}
	A \emph{Polish space} is a separable topological space which is
	metrizable by a complete metric. For any continuous map
	$f\colon X\to Y$ between Polish spaces $X$ and $Y$, the image of $f$
	is called an \emph{analytic topological space}.  For an analytic
	topological space $(X,\mathcal{O}_X)$, the measurable space
	$(X,\sigma(\mathcal{O}_X))$ is called an \emph{analytic measurable
	space}.
\end{definition}

Let us review the two key facts they used in~\cite{ClercFKP19}.
The first one is the following ``elegant Borel space analogue of the Stone--Weierstrass theorem''~\cite{Arveson76}.

\begin{fact}[{Unique Structure Theorem~\cite[Thm.~3.3.5]{Arveson76}}]
	\label{fact:exMarkovProcessesUST}
	Let $X\in\Meas$ be an analytic measurable space and $\mathscr{E}$ be
	an (at most) countable family of measurable subsets of $X$ such that
	$X\in\mathscr{E}$.  Define an equivalence relation
	$\equiv_{\mathscr{E}}$ by
	\[ x\equiv_{\mathscr{E}}y \iff \forall S\in\mathscr{E},(x\in S \iff
	y\in S).
	\]
	
If $S\subseteq X$ is measurable and $\equiv_{\mathscr{E}}$-closed,
	then $S\in\sigma(\mathscr{E})$.
\end{fact}

In the fact above, we use the operations of $\sigma$-algebras to construct $S$.
The second key fact is about ``decomposing'' those operations into two parts.

\begin{definition}
	Let $X$ be a set.  A family of subsets of $X$ is called a
	\emph{$\pi$-system} if it is closed under finite intersections.  A
	family of subsets of $X$ is a \emph{$\lambda$-system} if
	it is closed under complement and countable disjoint unions.
\end{definition}

Intuitively, $\pi$-systems correspond to the propositional connectives of $\esit_{\mathrm{CFKP}}$ and $\lambda$-systems correspond to ``approximation.''
These two operations are enough to recover all $\sigma$-algebra operations:

\begin{fact}[$\pi$-$\lambda$ Theorem~\cite{Dynkin60}]
	\label{fact:exMarkovProcessesPiLambda}
	If $\Pi$ is a $\pi$-system, $\Lambda$ is a $\lambda$-system, and
	$\Pi\subseteq\Lambda$, then $\sigma(\Pi)\subseteq\Lambda$.
\end{fact}

Using \cref{fact:exMarkovProcessesUST,fact:exMarkovProcessesPiLambda}, we obtain a sufficient condition for being an approximating family.
The proof follows the two steps outlined in \cref{rem:twoStepsInApprox}: 1) we can approximate a given $h\colon X\to 2$ by $\sigma$-algebra operations (\cref{fact:exMarkovProcessesUST}), which can be reduced to $\lambda$-system operations (\cref{fact:exMarkovProcessesPiLambda}); and 2) $\lambda$-system operations are in some sense ``preserved'' by the modalities (since measures are $\sigma$-additive).
\begin{propositionrep}
	\label{prop:exMarkovProcessesApproximability}
	Assume the setting of \cref{def:exMarkovProcessesSetting}.
	Let $X\in\Meas$.
	A subset $S\subseteq\Meas(X,2)$ is approximating if the following hold:
	\begin{itemize}
		\item $X$ is an analytic measurable space.
		\item $S$ is at most countable.
		\item For $k,l\in S$, $\top$ and $k\wedge l$ are also included in $S$.\myqed
	\end{itemize}
\end{propositionrep}
\begin{proof}	
	Fix any $h\colon \bigsqcap_{k\in S}k^*\Eq_2\to\Eq_2$, any $a\in A$, and any $r\in\mathbb{Q}\cap[0,1]$.
	By definition, it suffices to show \[
		\bigsqcap_{k\in S, a'\in A, r'\in\mathbb{Q}\cap[0,1]} (\tau_{a',r'}\co Bk)^*\Eq_2\sqsubseteq (\tau_{a,r}\co B(ph))^*\Eq_2.
	\]

	First we concretize these formulas.
	Let $R=\bigsqcap_{k\in S, a'\in A, r'\in\mathbb{Q}\cap[0,1]} (\tau_{a',r'}\co Bk)^*\Eq_2$.
	Using the definition of the arrow part of the functor $B=(\SGiry -)^A$, the relation $R$ on $(\SGiry X)^A$ can be rephrased as
	\begin{align*}
		&((\mu_a)_{a\in A},(\nu_a)_{a\in A})\in R \\
		&\iff\forall k,a',r',(\mu_{a'}(k^{-1}(\{1\}))>r'\iff\nu_{a'}(k^{-1}(\{1\}))>r') \\
		&\iff\forall k,a',r',(\mu_{a'}(k^{-1}(\{1\}))=\nu_{a'}(k^{-1}(\{1\}))),
	\end{align*}
	where $k\in S$, $a'\in A$, and $r'\in\mathbb{Q}\cap[0,1]$.
	In the same way, we can concretely describe $R'=(\tau_{a,r}\co B(ph))^*\Eq_2$ as
	\begin{align*}
		&((\mu_a)_{a\in A},(\nu_a)_{a\in A})\in R' \\
		&\iff(\mu_{a}((ph)^{-1}(\{1\}))>r\iff\nu_{a}((ph)^{-1}(\{1\}))>r).
	\end{align*}
	Thus, it suffices to show that the set $Y'=\{(\mu_a)_{a\in A}~|~\mu_{a}((ph)^{-1}(\{1\}))>r\}\subseteq (\SGiry X)^A$ is $R$-closed.
	Let $X'=(ph)^{-1}(\{1\})\subseteq X$.
	Now $Y'=\{(\mu_a)_{a\in A}~|~\mu_{a}(X')>r\}$.
	
	The set $X'$ corresponds to $h$, and we will ``approximate'' this by sets corresponding to the elements of $S$.
	Let $\mathscr{E}=\{k^{-1}(\{1\})~|~k\in S\}$ and define an equivalence relation
	$\equiv_{\mathscr{E}}$ by
	\[
	x\equiv_{\mathscr{E}}y \iff \forall E\in\mathscr{E},(x\in E \iff y\in E).
	\]
	Since $\equiv_{\mathscr{E}}$ coincides with the meet $\bigsqcap_{k\in S}k^*\Eq_2\in(\EqRelMeas)_X$, $X'$ is $\equiv_{\mathscr{E}}$-closed.
	Since $X$ is analytic and $S$ is at most countable, we can apply \cref{fact:exMarkovProcessesUST} and show $X'\in\sigma(\mathscr{E})$.
	
	Now we show that $Y'$ is $R$-closed.
	In this step, intuitively, we use the fact that the modalities ``preserve'' the ``approximation'' by the operation of $\Lambda$-system.
	Assume $((\mu_a)_{a\in A},(\nu_a)_{a\in A})\in R$ and
	$(\mu_a)_{a\in A}\in Y'$.
	Define a family $\Lambda$ of measurable subsets of $X$ by
	\begin{displaymath}
		E\in\Lambda \iff
		\forall a\in A, \mu_a(E) = \nu_a(E).
	\end{displaymath}
	Since $S$ is closed under $\top$ and $\wedge$, $\mathscr{E}$ is a
	$\pi$-system.
	On the other hand, by the definition of measure, $\Lambda$ is a $\lambda$-system.
	Since $((\mu_a)_{a\in A},(\nu_a)_{a\in A})\in R$, $\mathscr{E}\subseteq\Lambda$. \Cref{fact:exMarkovProcessesPiLambda} implies
	$\sigma(\mathscr{E})\subseteq\Lambda$.  In particular,
	$X'\in\Lambda$. This and $(\mu_a)_{a\in A}\in Y'$ imply
	$(\nu_a)_{a\in A}\in Y'$.
\end{proof}

From this proposition and \cref{thm:knasterTarskiApproxPrinciple}, we obtain the following expressivity result:
\begin{corollaryrep}
	Let $x\colon X\to(\SGiry X)^A$ be an LMP and $s,t\in X$ its states.
	Assume that the label set $A$ is at most countable and that $X$ is an analytic measurable space.
	
	Then $\esit_{\mathrm{CFKP}}$ is expressive for $x$ (\cref{def:expressive}):
	that is, If $s\vDash\varphi\iff t\vDash\varphi$ holds for every $\varphi\in L_{\esit_{\mathrm{CFKP}}}$, then $s$ and $t$ are probabilistically bisimilar.\myqed
\end{corollaryrep}
\begin{proof}
	Since $A$ is at most countable, $\{\sem{\varphi}~|~\varphi\in L_{\esit_{\mathrm{CFKP}}}\}\subseteq\Meas(X,2)$ is also at most countable.
	Moreover, since the logic has $\top$ and $\wedge$, $\{\sem{\varphi}~|~\varphi\in L_{\esit_{\mathrm{CFKP}}}\}\subseteq\Meas(X,2)$ is closed under these operations.
	Thus we can use \cref{prop:exMarkovProcessesApproximability} and \cref{thm:knasterTarskiApproxPrinciple}.
\end{proof}

\section{Expressivity for the Bisimulation Uniformity}
\label{sec:exBisimUnif}

In this section, we introduce \emph{bisimulation uniformity} as a coinductive predicate in a fibration and a logic for it.
By using our main results and a known mathematical result analogous to the Stone--Weierstrass theorem, the logic is readily proved to be adequate and expressive w.r.t.~bisimulation uniformity.
This example shows how our abstract framework can help to explore new bisimilarity-like notions.

\subsection{Uniform Structure as Fibrational Predicate}

Topological space can be regarded as an abstraction of (pseudo-)metric spaces w.r.t.~continuous maps.
In much the same way, \emph{uniform space}~\cite{BourbakiUniformStructures} is an abstraction of (pseudo-)metric spaces w.r.t.~uniformly continuous maps.

\begin{definition}[{from~\cite[Def.~1]{BourbakiUniformStructures}}]
	A \emph{uniform structure}, or \emph{uniformity}, on a set $X$ is a nonempty family $\mathscr{U}\subseteq\Pow(X\times X)$ of subsets of $X\times X$ satisfying the following:
	\begin{itemize}
		\item If $V\in\mathscr{U}$ and $V\subseteq V'\subseteq X\times X$, then $V'\in\mathscr{U}$.
		\item If $V,W\in\mathscr{U}$, then $V\cap W\in\mathscr{U}$.
		\item If $V\in\mathscr{U}$, then $\{(x,x)~|~x\in X\}\subseteq V$.
		\item If $V\in\mathscr{U}$, then $\{(y,x)~|~(x,y)\in V\}\in\mathscr{U}$.
		\item If $V\in\mathscr{U}$, then there exists $W\in\mathscr{U}$ such that $\{(x,z)~|~\exists y~(x,y)\in W\wedge(y,z)\in W\}\subseteq V$.
	\end{itemize}
	Here each element $V\in\mathscr{U}$ is called an \emph{entourage}.
	A pair $(X,\mathscr{U})$ of a set and a uniformity on it is called a \emph{uniform space}.
	
	A function $f\colon X\to Y$ is a \emph{uniformly continuous map from $(X,\mathscr{U}_X)$ to $(Y,\mathscr{U}_Y)$} if, for each entourage $V\in\mathscr{U}_Y$, $\{(x,x')~|~(f(x),f(x'))\in V\}\subseteq X\times X$ is an enrourage of $(X,\mathscr{U}_X)$.
	The category of uniform spaces and uniformly continuous maps is denoted $\Unif$.
\end{definition}

Each entourage represents some degree of ``closeness.''
The following example is an archetypal one:
\begin{example}
	\label{ex:uniformSpace}
	Let $(X,d)$ be a pseudometric space.
	Define a family $\mathscr{U}\subseteq\Pow(X\times X)$ as the set of all relations of the form $\{(x,x')~|~d(x,x')<\varepsilon\}$ for $\varepsilon>0$ and their supersets.
	Then $(X,\mathscr{U})$ is a uniform space.
\end{example}

Some of the concepts considered for metric spaces, like completion, total boundedness, and characterization of compactness, can be lifted to uniform spaces.
For us, the most important fact is that they form a $\CLatw$-fibration:
\begin{proposition}[{from~\cite[Propositions 4 and 5]{BourbakiUniformStructures}}]
	\label{prop:exBisimUnifCLatwFib}
	The forgetful functor $\Unif\to\Set$ is a $\CLatw$-fibration.\myqed
\end{proposition}

Thus we can use uniform structures as a sort of indistinguishability structure.
A uniform structure on a finite set is essentially the same as an equivalence relation.
For infinite sets, however, it can be a helpful way to analyze coalgebras that is more quantitative than an equivalence relation and more robust than a pseudometric.

\subsection{Expressivity Situation for Bisimulation Uniformity}

\begin{definition}
	\label{def:exBisimUnifSetting}
	Define an expressivity situation $\esit_{\mathrm{BU}}$ by:
	\begin{itemize}
		\item Its fibration is $\Unif\to\Set$ (\cref{prop:exBisimUnifCLatwFib}).
		\item Its truth-value object is $\RR$ and its observation predicate is $\mathscr{U}_e$, the uniformity defined using the usual Euclidean metric as in~\cref{ex:uniformSpace}.
		\item The ranked alphabet of its propositional connectives is $\Sigma=\{1^0,\min^2\}\cup\{(r+)^1,(r\times)^1~|~r\in\RR\}$. Its propositional structure $(f_\sigma\colon\RR^{\rank(\sigma)}\to\RR)_{\sigma\in\Sigma}$ is specified by:
		\begin{align*}
			f_{1}()&=1&f_{\min}(x,y)&=\min(x,y)\\
			f_{r+}(x)&=r+x&f_{r\times}(x)&=rx
		\end{align*}
		\item The behavior functor $B\colon\Set\to\Set$, the set of its modality indices $\Lambda$, and its observation modalities $(\tau_\lambda\colon B\RR\to\RR)_{\lambda\in\Lambda}$ are arbitrary.
	\end{itemize}
\end{definition}

For a $B$-coalgebra $x\colon X\to BX$, the codensity lifting $\codlif{B}{\mathscr{U}_e}{\tau}$ (\cref{def:codensityLifting}) yields the codensity bisimilarity $\CBisim{\mathscr{U}_e}{\tau}{x}\in\Unif_X$ (\cref{def:codensityBisim}), which is a uniformity on the set $X$.
We call it the \emph{bisimulation uniformity of $x$}.

On the other hand, the logic $L_{\esit_{\mathrm{BU}}}$ induces the fibrational logical equivalence $\LEq{\esit_{\mathrm{BU}}}{x}$ (\cref{def:logicalEq}) for each $x\colon X\to BX$.
We call this the \emph{logical uniformity of $x$}.

By \cref{prop:adequacy}, we obtain the following:
\begin{proposition}
	Assume the setting of \cref{def:exBisimUnifSetting}.
	Let $x\colon X\to BX$ be a $B$-coalgebra.
	Any entourage of the logical uniformity is also an entourage of the bisimulation uniformity.
	In particular, for any $\varphi\in L_{\esit_{\mathrm{BU}}}$, $\sem{\varphi}_x\colon X\to \RR$ is uniformly continuous w.r.t.~the bisimulation uniformity.\myqed
\end{proposition}

To prove expressivity, we have to make further assumptions.
\begin{assumption}
	\label{asm:exBisimUnif}
	For $\esit_{\mathrm{BU}}$, assume the following:
	\begin{enumerate}
		\item\label{item:exBisimUnifAsmBounded} If $k\colon X\to\RR$ is bounded, then $\tau_\lambda\co Bk\colon BX\to\RR$ is also bounded for each $\lambda\in\Lambda$.
		\item\label{item:exBisimUnifAsmConverge} If a sequence $k_i$ of functions of type $X\arrow\RR$ uniformly converges into $h$,
		then $\tau_\lambda\co Bk_i\colon BX\arrow\RR$ uniformly converges into $\tau_\lambda\co Bh$ for each $\lambda\in\Lambda$.
	\end{enumerate}
\end{assumption}

The key in the expressivity proof is the following known Stone--Weierstrass-like result:
\begin{fact}[{from~\cite[Thm.~1]{CsaszarApproximation1971}}]
	\label{fact:exBisimUnif}
	Let $X$ be a set and $\Gamma\subseteq\RR$ a set of real numbers unbounded both from above and below.
	Assume a family $\Phi$ of bounded real-valued function satisfies the following:
	\begin{enumerate}
		\item\label{item:exBisimUnifFactConst} Every constant is in $\Phi$.
		\item\label{item:exBisimUnifFactScaling} For $f\in\Phi$ and $r\in\Gamma$, $rf\in\Phi$ holds.
		\item\label{item:exBisimUnifFactShift} For $f\in\Phi$ and $r\in\RR$, $r+f\in\Phi$ holds.
		\item\label{item:exBisimUnifFactMinMax} For $f,g\in\Phi$, $\min(f,g),\max(f,g)\in\Phi$ holds.
	\end{enumerate}
	Let $\mathscr{U}_\Phi$ be the coarsest uniformity on $X$ that makes every function in $\Phi$ uniformly continuous.
	Then any real-valued function uniformly continuous w.r.t.~$\mathscr{U}_\Phi$ is the limit of a uniformly convergent sequence of elements of $\Phi$.
\end{fact}

By using this, we can show that a suitable set is approximating. In its proof, we follow the two steps discussed in~\cref{rem:twoStepsInApprox}.

\begin{propositionrep}
	\label{prop:exBisimUnifApproximability}
	Assume the setting of \cref{def:exBisimUnifSetting}.
	Let $X\in\Set$.
	Under \cref{asm:exBisimUnif}, a subset $S\subseteq\Set(X,\RR)$ is approximating if the following hold:
	\begin{itemize}
		\item Every function in $S$ is bounded.
		\item $1\in S$.
		\item $S$ is closed under the three operations $\min$, $(r+)$, and $(r\times)$ for every
		$r\in\RR$.\myqed
	\end{itemize}
\end{propositionrep}
\begin{proof}
	Define a uniformity $\mathscr{U}_S$ as the coarsest uniformity that every $k\in S$ is a uniformly continuous map $(X,\mathscr{U}_S)\to(\RR,\mathscr{U}_e)$.
	Fix $h\colon (X,\mathscr{U}_S) \to (\RR,\mathscr{U}_e)$ and $\lambda\in\Lambda$.
	We show that, in the fiber $\Unif_X$, \[
		\bigsqcap_{k\in S, \lambda'\in\Lambda} (\tau_{\lambda'}\co Bk)^*\mathscr{U}_e\sqsubseteq (\tau_\lambda\co B(ph))^*\mathscr{U}_e
	\] holds.
	Since $\{\{(x,y)\in\RR^2~|~d_e(x,y)<\varepsilon\}~|~\varepsilon>0\}$ is a fundamental system of entourages of $\mathscr{U}_e$, the family $\{\{(x,y)\in X^2~|~d_e((\tau_\lambda\co B(ph))(x),(\tau_\lambda\co B(ph))(y))<\varepsilon\}~|~\varepsilon>0\}$ is a fundamental system of entourages of $(\tau_\lambda\co B(ph))^*\mathscr{U}_e$.
	So it suffices to show that each relation in the family is an entourage of $\bigsqcap_{k\in S, \lambda'\in\Lambda} (\tau_{\lambda'}\co Bk)^*\mathscr{U}_e$.
	
	We show the following stronger claim:
	\begin{claim}
		For any $\varepsilon>0$, there exists $k\in S$ such that $\{(x,y)\in X^2~|~d_e((\tau_\lambda\co Bk)(x),(\tau_\lambda\co Bk)(y))<\varepsilon/3\}$ is a subset of $\{(x,y)\in X^2~|~d_e((\tau_\lambda\co B(ph))(x),(\tau_\lambda\co B(ph))(y))<\varepsilon\}$.
	\end{claim}
	
	Use \cref{fact:exBisimUnif} for $\Gamma=\RR$ and $\Phi=S$.
	The condition~(\ref{item:exBisimUnifFactConst}) is satisfied because every constant is a multiple of $1$, and the condition~(\ref{item:exBisimUnifFactMinMax}) is satisfied because $\max(x,y)=-\min(-x,-y)$.
	This ensures the existence of a sequence $(k_n\colon X\to\RR)_{n=1,2,\dots}$ that uniformly converges
	to $ph$ as $n\to\infty$.
	
	Fix $\varepsilon>0$.
	By the assumption~(\ref{item:exBisimUnifAsmConverge}), the sequence $(\tau_\lambda\co Bk_n)_{n=1,2,\dots}$ also uniformly converges to $\tau_\lambda\co B(ph)$ as $n\to\infty$.
	Thus, we can fix $n$ so that, for each $x\in X$, $d_e(\tau_{\lambda}((B(ph))(x)),\tau_{\lambda}((Bk_n)(x)))<\varepsilon/3$ holds.
	From the triangle inequality, we can take $k=pk_n$ for the claim above to hold.
\end{proof}

From this, we can obtain expressivity:

\begin{corollaryrep}
	Assume the setting of \cref{def:exBisimUnifSetting}.
	Let $x\colon X\to BX$ be a $B$-coalgebra.
	Under \cref{asm:exBisimUnif}, the bisimulation uniformity coincides with the logical uniformity, i.e., the former is characterized as the coarsest uniformity making every $\sem{\varphi}_x\colon X\to\RR$ uniformly continuous.\myqed
\end{corollaryrep}
\begin{proof}
	Use \cref{thm:knasterTarskiApproxPrinciple}.
	Indeed, by the assumption~(\ref{item:exBisimUnifAsmBounded}), $\sem{\varphi}_x$ is bounded for every $\varphi\in L_{\esit_{\mathrm{BU}}}$.
\end{proof}

\newcommand{\mt}[1]{\widetilde{#1}} \newcommand{\mc}[1]{{\mathscr #1}}
\newcommand{\modext}[2]{S_{#1+#2}}
\newcommand{\rloop}[2][-]{\save \POS!R(.7) \ar@(ru,rd)^#1{#2} \restore}
\newcommand{\lloop}[2][-]{\save \POS!L(.7) \ar@(lu,ld)_#1{#2} \restore}
\newcommand{\Algo}[1]{\mathcal{#1}}
\newcommand{\teq}{\mathbin{\triangleq}}
\newcommand{\blank}{\place}
\newcommand{\KR}{\mathrm{MLCP}}
\newcommand{\rk}[1]{{\rank(#1)}}
\renewcommand{\angle}[1]{\langle #1\rangle}
\newcommand{\ctuple}[1]{\angle{#1}}
\newcommand{\ptuple}[2]{\angle{#1}_{#2}}

\newcommand{\ca}[2]{\textcolor{black}{#2}}
\newcommand{\coalgebra}[3]{\textcolor{black}{coalgebra ${#3}:{#2}\to {#1}{#2}$}}
\newcommand{\beacoalgebra}[3]{\textcolor{black}{${#3}:{#2}\to {#1}{#2}$ be a coalgebra}}

\section{An Abstract Look at Expressivity and Approximation}
\label{sec:dualityAndComma}

Lastly, we take an abstract approach to the concept of expressivity
and approximation by combining two studies on coalgebraic modal
logic: fibrational formulation of adequacy and
expressivity \cite{KR20}, and Klin's reformulation of duality-based
modal logic using comma categories \cite{Klin10}. The combination
leads us to a new look at adequacy and expressivity as a comparison
problem of final coalgebras through a functor. An approximating family is then
a key construct to solve this comparison problem.

\subsection{Fibrational Theory of Adequacy and Expressivity}

In \cite{KR20}, the third and fourth authors integrated duality-based
modal logic \cite{PavlovicMW06} and fibrational theory of bisimulation
\cite{HJ98}, and formulated adequacy and expressivity of modal
logic. Their formulation is built upon the following categorical
situation \cite[Asm. 14, Def. 15]{KR20}, which we tentatively call a
{\em modal logic with coinductive predicates} $\mc K$:
\begin{equation}
  \label{eq:exsitu}
  \vcenter{\xymatrix{
      \EE \ar[d]_-p \lloop[(.2)]{\ol B} \\
      \CC \lloop[(.2)]B \adjunction{r}{P}{Q} & \DD^\op
      \ar[ul]_-{\ol Q}
      \rloop[(.2)]{L^\op} \restore
    }}
  \quad
  \begin{aligned}
    +&\left(\begin{aligned}
        & \alpha:L\Phi\to\Phi~\text{initial}\\
        & \delta:B\circ Q\to Q\circ L^\op
      \end{aligned}\right) \\
    \text{s.t.}& \left(\begin{aligned}
        & p:\text{$\CLatw$-fibration} \\
        & \ol B: \text{lifting of $B$} \\
        & p\circ\ol Q=Q
      \end{aligned}\right) \\
  \end{aligned}
\end{equation}
The tuple $\mc L\triangleq(B,P\dashv Q,L,\delta,\alpha)$ is called a
duality-based modal logic, while the tuple $(p,B,\ol B)$ determines a
setting for fibrational bisimulation
(\cref{subsec:prelimcoinductionInFib}). Now let \beacoalgebra BXx.  Following \cite[Def. 15]{KR20}, $\mc K$ is said to be
\begin{itemize}
\item {\em adequate} for $x$ if $\nu(x^*\circ\ol B)\sqsubseteq \theory_{\ca Xx}^*(\ol Q\Phi)$, and
\item {\em expressive} for $x$ if $\nu(x^*\circ\ol B)\sqsupseteq \theory_{\ca Xx}^*(\ol Q\Phi)$.
\end{itemize}
Here, $\theory_{\ca Xx}:X\to Q\Phi$ is the theory morphism induced by $x$. It
corresponds to the interpretation function $\sem\blank_{\ca Xx}$ in
\cref{def:formulaSemantics}; see \cite{PavlovicMW06} for detail.

One might wonder how the above definition of adequacy and
expressivity is related to ours in \cref{def:expressive}. The
following construction establishes a formal connection:
\begin{theoremrep}
  For any expressivity
  situation $\esit=(p,B,\Omega,\bOmega,\cdots,\tau)$ whose base category $\CC$ has small powers and
  equalizers, there is a modal logic with  coinductive predicates $\KR(\esit)$ such that for any
  \coalgebra BXx, $\nu(x^*\co\ol B)=\CBisim{\bOmega}{\tau}{\ca Xx}$ and
  $\theory_{\ca Xx}^*(\ol Q\Phi)=\LEq{\esit}{\ca Xx}$ holds.  \myqed
\end{theoremrep}
\begin{proof}
  Let $\esit=(p,B,\Omega,\bOmega,\Sigma,\Lambda,(f_\sigma)_{\sigma\in\Sigma},(\tau_\lambda)_{\lambda\in\Lambda})$
  be an expressivity situation.
\newcommand{\Sig}{{S_\Sigma}} We write $\Sig:\Set\to\Set$ for the
  signature functor of $\Sigma$ defined by
  $\Sig X=\coprod_{\sigma\in\Sigma} X^{\rk\sigma}$.  An element of
  $\Sig X$ is represented as the form
  $\sigma(v_1,\cdots,v_{\rk\sigma})$ where $\sigma\in\Sigma$ and
  $v_1,\cdots,v_{\rk\sigma}\in X$.

  We aim to construct the following modal logic with coinductive predicates $\KR(\esit)$.
  \begin{equation}
    \label{eq:exsitu-esit}
    \vcenter{\xymatrix{
        \EE \ar[d]_-p \lloop[(.2)]{\codlif B\bOmega\tau} \\
        \CC \lloop[(.2)]B \adjunction{r}{P}{Q} & (\Set^\Sig)^\op
        \rloop[(.2)]{L^\op} \restore
        \ar[lu]_-{\ol Q}
      }}
    \quad
    \begin{aligned}
      +&\left(\begin{aligned}
          & \alpha:L\Phi\to\Phi~\text{initial}\\
          & \delta:B\circ Q\to Q\circ L^\op
        \end{aligned}\right) \\
      \text{s.t.}& \left(\begin{aligned}
          & p:\text{$\CLatw$-fibration} \\
          & \text{$\codlif B\bOmega\tau$ lifting of $B$} \\
          & p\circ\ol Q=Q
        \end{aligned}\right) \\
    \end{aligned}
  \end{equation}
  The remaining data $P,Q,\ol Q,\alpha,\delta$ are defined below.  In
  fact, $Q,\ol Q$ are constructed by the following common method.
  \setcounter{theorem}{3}
  \begin{lemma}\label{lem:adj}
    Let $\CC$ be a category with small powers and equalizers, and
    $(\Omega,f)$ be a $\Sigma$-algebra in $\CC$. Define a functor
    $P:\CC\to(\Set^\Sig)^\op$ by
    \begin{align*}
      PX&=(\CC(X,\Omega),f^P) & \text{where}~f^P(\sigma(g_1,\cdots,g_{\rk\sigma})) &=f_\sigma\circ\ctuple{g_1,\cdots,g_{\rk\sigma}}\\
      Ph(k)&=k\circ h.
    \end{align*}
    Then $P$ has a right adjoint $Q$.
  \end{lemma}
  \begin{nestedproof}
    Let $(X,x)\in\Set^\Sig$. We define an object $Q(X,x)$ to be the
    following equalizer in $\CC$:
    \begin{align*}
      &    \xymatrix@C=2cm{
        Q(X,x) \ar[r]^-{e_{X,x}^\CC} & X\pitchfork\Omega \ar@<.2pc>[r]^-{a_{X,x}^\CC} \ar@<-.2pc>[r]_-{b_{X,x}^\CC} & \Sigma X\pitchfork\Omega
                                                                                                                      }\\
      &a_{X,x}^\CC=\ptuple{ f_\sigma\circ\tuple{\pi_{v_1},\cdots,\pi_{v_{\rk\sigma}}}}{\sigma(v_1,\cdots,v_{\rk\sigma})\in \Sig X}\\
      &b_{X,x}^\CC=\ptuple{ \pi_{x(\sigma(v_1,\cdots,v_{\rk\sigma}))}}{\sigma(v_1,\cdots,v_{\rk\sigma})\in \Sig X}.
    \end{align*}
    The reason why we annotate the category $\CC$ on morphisms $a,b,e$
    is that we later consider $a,b,e$ in different categories. When
    $\CC$ is evident, we omit writing it.
    
    We define a morphism $\epsilon_{X,x}:(X,x)\to PQ(X,x)$ in
    $\Set^\Sig$ by $ \epsilon_{X,x}(v)=\pi_v\circ e_{X,x}$. It is
    indeed a morphism in $\Set^\Sig$ as:
    \begin{align*}
      f^P\circ\Sig\epsilon_{X,x}(\sigma(v_1,\cdots,v_{\rk\sigma}))
      &=f^P(\sigma(\epsilon_{X,x}(v_1),\cdots,\epsilon_{X,x}(v_{\rk\sigma})))\\
      &=f_\sigma\circ\ctuple{\pi_{v_1},\cdots,\pi_{v_{\rk\sigma}}}\circ e_{X,x}\\
      &=\pi_{x(\sigma(v_1,\cdots,v_{\rk\sigma}))}\circ e_{X,x}\\
      &=\epsilon_{X,x}\circ x(\sigma(v_1,\cdots,v_{\rk\sigma})).
    \end{align*}

    We show that $\epsilon_{X,x}$ is a universal arrow from
    $(X,x)\in\Set^\Sig$ to $P$.  Let $h:(X,x)\to PY$ be a morphism in
    $\Set^\Sig$. That is, $h:X\to\CC(Y,\Omega)$ is a function such
    that $f^P\circ\Sig h=h\circ x$. Then
    $\ptuple{h(v)}{v\in X}:Y\to X\pitchfork\Omega$ equalizes $a_{X,x}$
    and $b_{X,x}$:
    \begin{align*}
      a_{X,x} \circ \ptuple{h(v)}{v\in X}
      &=\ptuple{ f_\sigma\circ\tuple{h(v_1),\cdots,h(v_{\rk\sigma})}}{\sigma(v_1,\cdots,v_{\rk\sigma})\in \Sig X}\\
      &=\ptuple{f^P(\sigma(h(v_1),\cdots,h(v_{\rk\sigma})))}{\sigma(v_1,\cdots,v_{\rk\sigma})\in \Sig X}\\
      &=\ptuple{h(x(\sigma(v_1,\cdots,v_{\rk\sigma})))}{\sigma(v_1,\cdots,v_{\rk\sigma})\in \Sig X}\\
      &=b_{X,x}\circ\ptuple{h(v)}{v\in X}.
    \end{align*}
    We thus obtain the unique morphism $\mt h:Y\to Q(X,x)$ in $\EE$
    such that
    \begin{equation}
      \label{eq:univ}
      e_{X,x}\circ \mt h=\ptuple{h(v)}{v\in X}.
    \end{equation}
    This satisfies $P\mt h\circ\epsilon_{X,x}=h$ in $\Set^\Sig$
    because for any $v\in X$, we have
    \begin{displaymath}
      \epsilon_{X,x}(v)=\pi_v\circ e_{X,x}\circ\mt h=\pi_v\circ \ptuple{h(v)}{v\in X}=h(v).
    \end{displaymath}

    We show that such $\mt h$ is unique. Let $h':Y\to Q(X,x)$ be a
    morphism in $\EE$ such that $Ph'\circ\epsilon_{X,x}=h$. This means
    that $\pi_v\circ e_{X,x}\circ h'=h(v)$ holds for any $v\in X$,
    hence $\ptuple{h(v)}{v\in X}=e_{X,x}\circ h'$. From the universal
    property of the equalizer, $h'=\mt h$.
  \end{nestedproof}

  Since $\CC$ has small powers and equalizers, we define $P\dashv Q$
  in \cref{eq:exsitu-esit} to be the one arising from $\esit$'s
  $\Sigma$-algebra $(\Omega,f)$ in $\CC$ by Lemma \ref{lem:adj}.

  Since $p:\EE\to\CC$ is a $\CLatw$-fibration, $\EE$ also has small
  powers and equalizers that are strictly preserved by $p$. Moreover
  $\esit$ specifies another $\Sigma$-algebra $(\bOmega,g)$ in $\EE$
  above $(\Omega,f)$; see the fourth condition in \cref{def:expSitX}.
  Therefore we define $\ol Q$ in \cref{eq:exsitu-esit} to be the right
  adjoint arising from $(\bOmega,g)$ by Lemma \ref{lem:adj}. Since $p$
  preserves small powers and equalizers, $p\circ\ol Q=Q$ holds.

  We define $L$ in \cref{eq:exsitu-esit} to be the composite
  $F^\Sig\circ(\Lambda\times \place)\circ U^\Sig$, where
  $F^\Sig\dashv U^\Sig:\Set^\Sig\to\Set$ is the adjunction
  constructing free $\Sig$-algebras. Let us identify $L$-algebras:
  \begin{lemma}
    By $\Sigma+\Lambda$ we mean the ranked alphabet obtained by
    (disjointly) adding $\Lambda$-many rank-1 symbols to $\Sigma$. We
    have an isomorphism of categories:
    \begin{equation}
      \label{eq:isoalg}
      (\Set^\Sig)^L\cong\Set^{\modext\Sigma\Lambda}.
    \end{equation}
  \end{lemma}
  \begin{nestedproof}
    From the following bijection:
    \begin{displaymath}
      \Set^\Sig(L(X,x),(X,x))\cong
      \Set(\Lambda\times X,X)\cong
      \Set(X,X)^\Lambda,
    \end{displaymath}
    an $L$-algebra $l:L(X,x)\to (X,x)$ bijectively corresponds to a
    $\Lambda$-indexed family of endofunctions on $X$. Thus the tuple
    $(X,x,l)$ bijectively corresponds to an
    $\modext\Sigma\Lambda$-algebra on $X$.
  \end{nestedproof}

  We define an initial algebra $\alpha:L\Phi\to\Phi$ to be the one
  corresponding to an initial $\modext\Sigma\Lambda$-algebra
  constructed as the set $L_\esit$ of formulas of $\esit$ in
  (\cref{def:syntax}). Thus the carrier set of $\Phi\in\Set^\Sig$ is
  $L_\esit$.

  Under the right adjoint $Q$ and an endofunctor $L$ defined above, we
  obtain the following bijective correspondence:
  \begin{align*}
    \bec{}{[\Op{(\Set^\Sig)}, \CC] (B \circ Q, Q \circ \Op L)}{}
    \bec{\cong}{[\CC, \Op{(\Set^\Sig)}] (P\circ B, \Op L \circ P)}{adjoint mate}
    \bec{\cong}{[\Op{\Cat{C}}, \Set^\Sig] (L \circ P^\op, P^\op \circ \Op B)}{}
    \bec{=}{[\Op{\Cat{C}}, \Set^\Sig] (F^\Sig (\Lambda\times U^\Sig \place)\circ P^\op, P^\op
    \circ \Op B)}{by def.\ of $L$}
    \bec{\cong}{[\Op{\Cat{C}}, \Set] ((\Lambda\times U^\Sig \place)\circ P^\op, U^\Sig \circ P^\op \circ \Op B)}{}
    \bec{\cong}{[\Op{\Cat{C}}, \Set] (U^\Sig \circ P^\op, (\place)^{\Lambda} \circ U^\Sig \circ P^\op \circ \Op B)}{}
    \bec{=}{[\Op{\Cat{C}}, \Set] \bigl(\,\CC (\place, \Omega),  \bigl(\CC (B\place, \Omega)\bigr)^{\Lambda}\,\bigr)}{by def.\ of $P$}
    \becncr{\cong}{\bigl( \CC (B \Omega, \Omega)\bigr)^{\Lambda}}{by the Yoneda lemma.}
    \tag*{\qedhere}
  \end{align*}
  We therefore define $\delta$ to be the one corresponding to
  $\esit$'s modality
  $(\tau_\lambda)_{\lambda\in\Lambda}\in\CC(B\Omega,\Omega)^\Lambda$.

  Let \beacoalgebra BXx. It is ovbious that
  $\nu(x^*\co\ol B)=\CBisim{\bOmega}{\tau}{\ca Xx}$. We thus show that
  $\theory_{\ca Xx}^*(\ol Q\Phi)=\LEq{\esit}{\ca Xx}=\bigsqcap_{\phi\in
    L_\esit}\sem\phi_{\ca Xx}^*\bOmega$.  First, we construct the theory
  morphism $\theory_{\ca Xx}:X\to Q\Phi$.  From the inductive definition of the
  interpretation function $\sem\blank_{\ca Xx}$, it is a $\Set^\Sig$-algebra
  homomorphism of type $\Phi\to PX$.  Then we define $\theory_{\ca Xx}$ to be the
  adjoint mate $\mt{\sem\blank_{\ca Xx}}:X\to Q\Phi$. We note that
  $e_\Phi\circ \theory_{\ca Xx}=\ptuple{\sem\phi_{\ca Xx}}{\phi\in L_\esit}$ by
  \eqref{eq:univ}.

  Next, we explicitly compute $\ol Q\Phi$.  In $\EE$, the equalizer of
  $a^\EE_\Phi,b^\EE_\Phi:L_\esit\pitchfork\bOmega\to\Sig
  L_\esit\pitchfork\bOmega$ is given by a Cartesian lifting of the
  equalizing morphism of $pa^\EE_\Phi,pb^\EE_\Phi$ with
  $L_\esit\pitchfork\bOmega$.  Since $p$ strictly preserves powers, we
  have $pa^\EE_\Phi=a^\CC_\Phi,pb^\EE_\Phi=a^\CC_\Phi$, hence their
  equalizing morphism is $e_\Phi^\CC$. The right half of the following
  diagram describes this equalizing process.
  \begin{displaymath}
    \xymatrix@C=2.5cm{
      th^*(\ol Q\Phi) \ar@{.>}[r]^-{\ol{\theory_{\ca Xx}}(\ol Q\Phi)} & \ol Q\Phi \ar@{.>}[r]^-{e_{\Phi}^\EE=\ol{e_\Phi^\CC}(L_\esit\pitchfork\bOmega)} & L_\esit\pitchfork\bOmega \ar@<.2pc>[r]^-{a_{\Phi}^\EE} \ar@<-.2pc>[r]_-{b_{\Phi}^\EE} & \Sig L_\esit\pitchfork\bOmega & \EE \ar[d]^-p \\
      X \ar[r]^-{\theory_{\ca Xx}} & Q\Phi \ar[r]^-{e_{\Phi}^\CC} & L_\esit\pitchfork\Omega \ar@<.2pc>[r]^-{a_{\Phi}^\CC} \ar@<-.2pc>[r]_-{b_{\Phi}^\CC} & \Sig L_\esit\pitchfork\Omega & \CC
    }
  \end{displaymath}
  On the left half, we also add the Cartesian lifting of $\theory_{\ca Xx}$ with
  $\ol Q\Phi$, yielding the object $th^*(\ol Q\Phi)$.
  
  We finally remark that in the $\CLatw$-fibration $p:\EE\to\CC$, we
  have
  $L_\esit\pitchfork\bOmega=\bigsqcap_{\phi\in
    L_\esit}\pi_\phi^*\bOmega$.  By combining these facts, we conclude
  \begin{align*}
    th^*(\ol Q\Phi)
    =th^*(e_\Phi^*(L_\esit\pitchfork\bOmega))
    =th^*\left(e_\Phi^*\left(\bigsqcap_{\phi\in L_\esit}\pi_\phi^*\bOmega\right)\right)
    =\bigsqcap_{\phi\in L_\esit}(\pi_\phi\circ e_\Phi\circ th)^*\bOmega
    =\bigsqcap_{\phi\in L_\esit}\sem\phi_{\ca Xx}^*\bOmega.
  \end{align*}
  \qed
\end{proof}

\subsection{Fibration from Duality-Based Modal Logic}

In \cite{Klin10}, Klin studied duality-based modal logic using a comma
category (see \cite[Sect. II.6]{cwm2} for the definition). We quickly
review his study, reusing the duality-based modal logic $\mc L$ in
\eqref{eq:exsitu}. We consider the comma category
$\Id_{\CC}\downarrow Q$\footnote{ This is isomorphic to the category
  $P\downarrow\Id_{\Set^\Sigma}$ employed by Klin in \cite{Klin10}.  }
with the evident first projection functor
$\pi_1:\Id_{\CC}\downarrow Q\to\CC$. He showed the following results
that are relevant to us:
\begin{itemize}
\item \cite[Sect. 4]{Klin10} The natural transformation $\delta$
  determines a {\em lifting} $\Delta$ of $B$ along $\pi_1$, whose
  object part is given by $\Delta(X,Y,f)=(BX,LY,\delta_Y\circ Bf)$.
\item \cite[Corollary 4.4]{Klin10} The initial $L$-algebra $\alpha$
  induces a right adjoint right inverse $R$ of
  $\CA{\pi_1}:\CA{\Delta}\to\CA{B}$ (see
  \cref{prop:coindPredMoreAbstractly}). Its object part
  sends a \coalgebra BXx to
  $R\ca Xx=(X,\Phi,\theory_{\ca Xx})$.
\end{itemize}

We re-interpret these results in terms of the fibrational theory of
bisimulation. First, one can easily verify that
$\pi_1:\Id_\CC\downarrow Q\to\CC$ is a fibration and $\Delta$ is a
fibred lifting. Second, for any \coalgebra BXx, $R\ca Xx$ is a final object in the
fiber category $\CA{\pi_1}_{\ca Xx}$, which corresponds to a
$\Delta$-coinductive predicate $\nu(x^*\circ\Delta)$ over $x$ by
\cref{prop:coindPredMoreAbstractly}.
To summarize, a duality-based modal logic induces a setting for
fibrational bisimulation admitting coinductive predicates.

\subsection{Another Abstract Look at Adequacy and Expressivity}

We combine fibrational theory of adequacy and expressivity, and a
fibrational presentation of the duality-based modal logic based on
Klin's study.  The key factor connecting these two studies is the fibred
functor $H:\pi_1\to p$ given by $H(X,Y,f)=f^*(\ol QY)$.
\begin{equation}
  \label{eq:comparefib}
  \vcenter{
    \xymatrix@R=1em{
      \EE \ar[rd]_(.3)p \lloop{\ol B}
      & & \Id_\CC\downarrow Q \ar[ld]^-{\pi_1} \ar[ll]_-H \rloop\Delta \\
      & \CC \lloop B
    }
  }
\end{equation}
Using $H$, the definitions of adequacy and expressivity are equivalently
rewritten as follows. Let \beacoalgebra BXx.  Then the modal logic with coinductive predicates $\mc K$ is
\begin{itemize}
\item {\em adequate} for $x$ if
  $\nu(x^*\circ\ol B)\sqsubseteq H(\nu(x^*\circ\Delta))$, and
\item {\em expressive} for $x$ if
  $\nu(x^*\circ\ol B)\sqsupseteq H(\nu(x^*\circ\Delta))$.
\end{itemize}
That is, establishing adequacy and expressivity can be viewed as a
familiar problem of comparing final coalgebras in two (fibre)
categories connected by a functor.

\subsection{Another Abstract Look at Kleene Proof Principle}

From the above reformulation of expressivity (and adequacy), we easily
notice the following sound proof method.  Below we impose the
following conditions (*) on 
\begin{theoremrep}\label{thm:kleeneabs}
  Suppose that in the modal logic with coinductive predicates $\mc K$ \eqref{eq:exsitu},
  $\DD$ has an initial object, $L$-initial sequence stabilizes, and
  $\ol Q$ is a right adjoint.  Let \beacoalgebra BXx such that
  $x^*\circ\ol B$-final sequence stabilizes after $\omega$-steps. Then
  each of \eqref{eq:kleene} and \eqref{eq:comm} implies that $\mc{K}$ is expressive for $x$. 
  \begin{align}
    & \label{eq:kleene}
      \forall{i\in\omega}~.~H((x^*\circ\Delta)^i(\top))\sqsubseteq (x^*\circ\ol B)^i(\top) \\
    & \label{eq:comm}
      \forall{i\in\omega}~.~
      H\co\Delta\co (x^*\co \Delta)^i(\top)\sqsubseteq \ol B\co H\co (x^*\co \Delta)^i(\top)
  \end{align}
  \myqed
\end{theoremrep}
\begin{proof}
  The fibre category $(\CC\downarrow Q)_X$ of the fibration
  $\pi_1:\CC\downarrow Q\arrow\CC$ is isomorphic to the comma category
  $X\downarrow Q$. We therefore work with $X\downarrow Q$ instead.
  Then the endofunctor $x^*\circ\Delta$ over $(\CC\downarrow Q)_X$ is
  redefined on $X\downarrow Q$ by
  \begin{displaymath}
    x^*\circ\Delta(Y,f)=(LY,\delta\circ Bf\circ x).
  \end{displaymath}
  We also consider the evident forgetful functor
  $p:X\downarrow Q\arrow\DD^{op}$. This functor reflects isomorphisms,
  and $x^*\circ\Delta$ is a lifting of $L^{op}$ along $p$.

  First consider the final $x^*\circ\Delta$-sequence in
  $X\downarrow Q$:
  \begin{displaymath}
    \xymatrix{
      \top & x^*\circ\Delta \top \ar[l] & x^*\circ\Delta(x^*\circ\Delta \top) \ar[l] & \cdots \ar[l] \\
    }
  \end{displaymath}
  where $\top= (0,!:X\arrow Q0)$ is a terminal object.  Notice that
  $Q:\DD^{op}\arrow\CC$ is a right adjoint, hence $Q0$ is a terminal
  object (here $0$ is the assumed initial object in $\DD$).  The
  unique morphism from $(Y,f)$ to $(0,!)$ is $!:Y\arrow 0$ in
  $\DD^{op}$ (i.e. $!:0\arrow Y$ in $\DD$).

  The functor $p$ sends this final $x^*\circ\Delta$-sequence to the
  following diagram, which is again the final $L^{op}$-sequcence in
  $\DD^{op}$ (i.e. $L$-initial sequence in $\DD$):
  \begin{displaymath}
    \xymatrix{
      0 & L0 \ar[l]_-{L!} & L^20 \ar[l]_-{L^2!} & \cdots \ar[l]
    }
  \end{displaymath}
  From the assumption, this $L^{op}$-sequence converges, say at an
  ordinal $\lambda$. Since $p$ reflects isomorphisms, the final
  $x^*\circ\Delta$-sequence also converges at $\lambda$; without loss
  of gerality we assume $\omega\le\lambda$.  We write
  e$\nu(x^*\circ\Delta)$ for this converging object.  We take the
  evident morphism
  $m:\nu(x^*\circ\Delta)\arrow\lim_{i\in\omega}(x^*\circ\Delta)^i\top$.

  From the assumption,
  $H((x^*\circ\Delta)^i\top)\sqsubseteq (x^*\circ\ol B)^i\top$ holds
  for any $i\in\omega$, and the $\omega^{op}$-limit of the right hand
  side yields $\nu(x^*\circ\ol B)$. On the other hand, from $\ol Q$
  being a right adjoint, $H$ also becomes a right adjoint too (we
  postpone the proof of this claim). Therefore
  $H(\lim_{i\in\omega}(x^*\circ\Delta)^i\top)=\lim_{i\in\omega}
  H((x^*\circ\Delta)^i\top)$.  Therefore
  \begin{align*}
    \bec{}{H(\nu(x^*\circ\Delta))}{}
    \bec{\sqsubseteq}{H(\lim_{i\in\omega}(x^*\circ\Delta)^i\top)}{$H(m)$ witnesses this inequality}
    \bec{=}{\lim_{i\in\omega} H((x^*\circ\Delta)^i\top)}{$H$ preserves limits}
    \bec{\sqsubseteq}{\nu(x^*\circ\ol B).}{}
  \end{align*}
  Next, we show that \eqref{eq:comm} implies \eqref{eq:kleene} by induction.
  The base case is trivial. For the case $i+1\in\omega$,
  \begin{align*}
    \bec{}{H\co(x^*\co\Delta)^{i+1}(\top)}{}
    \bec{=}{H\co x^*\co\Delta\co(x^*\co\Delta)^{i}(\top)}{unfolding}
    \bec{=}{x^*\co H\co \Delta\co(x^*\co\Delta)^{i}(\top)}{$H$ being fibred}
    \bec{\sqsubseteq}{x^*\co\ol B\co H\co(x^*\co\Delta)^{i}(\top)}{\eqref{eq:comm}}
    \bec{\sqsubseteq}{(x^*\co\ol B)^{i+1}(\top).}{IH}
  \end{align*}

  We finally show that $Q$ being a right adjoint implies that $H$ is
  so too.  We write $\eta_X:X\to\ol Q\ol PX$ for the unit of the
  adjunction $\ol P\dashv\ol Q:\CC\to\EE$.  Let $X\in\EE$. Then we
  obtain an object $KX\triangleq(pX,\ol PX,p\eta_X)\in\CC\downarrow Q$, and we
  obtain a vertical morphism $X\le HKX=(p\eta_X)^*(\ol Q\ol PX)$,
  which we name $\dot\eta_X$.

  Now let $(Y,I,f)\in\CC\downarrow Q$ be an object
  and $h:X\to H(Y,I,f)=f^*\ol QI$ be a morphism in $\EE$.
  Then we obtain the composite
  \begin{displaymath}
    \ol f(\ol QI)\circ h:X\to\ol QI,
  \end{displaymath}
  were $\ol f(\ol QI)$ is the Cartesian lifting of $f$ with $\ol QI$.
  We then take the adjoint mate of this composite, and write it by
  $m:\ol PX\to I$. It is immediate that the pair $(ph,m)$ forms a
  morphism from $KX$ to $(Y,I,f)$ in $\CC\downarrow Q$. Moreover,
  $H(ph,m):HKX\to H(Y,I,f)$ is the unique morphism above $ph$ such
  that $h=H(ph,m)\circ \dot\eta_X$.

  Suppose that there is another $(a,b):KX\to(Y,f,I)$ such that
  $h=H(a,b)\circ\dot\eta_X$.  Since $H(a,b)$ is above $a$ and
  $\dot\eta_X$ is vertical, we have $ph=a$.
  Next,
  \begin{displaymath}
    \ol f(\ol QI)\circ h
    =
    \ol f(\ol QI)\circ H(a,b)\circ \dot\eta_X
    =
    \ol Qb\circ\ol{\eta_X}(\ol Q\ol PX)\circ \dot\eta_X
    =
    \ol Qb\circ\eta_X.
  \end{displaymath}
  Therefore $b=m$.
  \begin{displaymath}
    \xymatrix@C=2cm{
      X \ar[r]^-{\dot\eta_X} \ar[rd]_-h & (p\eta_X)^*(\ol Q\ol PX) \ar[r]^-{\ol{p\eta_X}(\ol Q\ol PX)} \ar[d]^-{H(a,b)}& \ol Q\ol PX \ar[d]^-{\ol Qb} \\
      & f^*\ol QI \ar[r]_-{\ol f(\ol QI)} & \ol QI \\
      & pX \ar[d]_-a \ar[r]^-{p\eta_X} & p\ol Q\ol PX=Q\ol PX \ar[d]^-{Qb} & \ol PX \ar[d]^-b \\
      & Y \ar[r]_-f & QI & I
    }
  \end{displaymath}
\end{proof}

In other words, \eqref{eq:kleene} says that the $H$-image of the
$x^*\circ\Delta$-final sequence is bound by the $x^*\circ\ol B$-final
sequence. The above theorem is working behind the proof of Kleene
proof principle for expressivity (\cref{thm:kleeneApproxPrinciple});
in the situation \eqref{eq:comparefib} arising from $\KR(\esit)$,
\cref{thm:kleeneApproxPrinciple} first shows \eqref{eq:kleene} using
the assumption of approximating family of observations, then concludes
the expressivity by invoking \cref{thm:kleeneabs}.  On the other hand,
\eqref{eq:comm} is an analogy of approximating family of
observations appeared in \cref{thm:kleeneApproxPrinciple} in the
abstract set-up \eqref{eq:comparefib}.

\section{Conclusions and Future Work}

We introduced a categorical framework to study expressivity of quantitative modal logics,
based on the novel notion of approximating family. This enabled us to cover not only existing examples (\cref{sec:exampleKoenigMM} and \cref{sec:exMarkovProcesses}) but also a new one (\cref{sec:exBisimUnif}).
We conclude with some future research directions.

\paragraph{Making Use of Size Restrictions on Functors}
Many existing expressivity results make use of size restriction condition on the behavior functor $B$.
For example,~\cite{HennessyM85} required \emph{image-finiteness},~\cite{schr08:expr} used \emph{$\kappa$-accessibility}, and~\cite{WildS20} was based on a quantitative notion, \emph{finitary separability}.
Importing these size restrictions is future work.
A starting point can be~\cite{HasuoKC18}, which successfully connected the finitarity of the behavior functor and the length of the final chain in a fiber.

\paragraph{Study of Bisimulation Uniformity}
We defined bisimulation uniformity in~\cref{sec:exBisimUnif}, but there are many topics left to study.
One primary subject is the connection to bisimilarity and bisimulation metric.
It is also important to see if it is robust under parameter changes of the target system.

\paragraph{Seeking Stone--Weierstrass-like Theorems}
To use our framework to show expressivity, one has to obtain a sufficient condition for being an approximating family.
In many cases, this is reduced to finding an appropriate ``Stone--Weierstrass-like'' theorem.
Concretely find ones and apply them to modal logics (other than those we have mentioned) is future work.
Another research direction is to a seek connection to~\cite{Hofmann02}, where ``Stone--Weierstrass-like'' theorems are formulated in another way.

\bibliographystyle{IEEEtran}
\bibliography{refs}

\clearpage
\onecolumn
\appendix
\renewcommand{\appendix}{}
\renewcommand{\myqed}{}
\renewcommand{\mymyqed}{}

\subsection{Further on~\cref{rem:chainLength}}
\label{appendix:remchainLength}
	Assume that $L_{\esit}$ is expressive for $x\colon X\to BX$.
	First, in much the same way as \cref{prop:adequacy}, we can show ``stepwise adequacy'': \[
	(x^*\co\codlif{B}{\bOmega}{\tau})^n(\top)\sqsubseteq\bigsqcap_{\varphi\in L_{\esit},\depth(\varphi)\le n}\sem{\varphi}_{x}^{*}\bOmega
	\] holds for $n\in\omega$.
	Taking the meets of both sides for $n\in\omega$ shows \[
	\bigsqcap_{n\in\omega}(x^*\co\codlif{B}{\bOmega}{\tau})^n(\top)\sqsubseteq\bigsqcap_{\varphi\in L_{\esit}}\sem{\varphi}_{x}^{*}\bOmega=\LEq{\esit}{x}\sqsubseteq\CBisim{\bOmega}{\tau}{x}.
	\]

\subsection{More on Total Boundedness}
\label{subsubsec:totalBoundedness}

For the arguments in \cref{sec:exampleKoenigMM}, finiteness of $\Lambda$ is crucial, which was not very obvious in~\cite{KonigM18}.
Here we consider a handy counterexample against \cref{fact:exampleKoenigMM} where $\Lambda$ is infinite.
It turns out that the counterexample also affects our approximation argument.

Define $d_2\colon 2\times 2\to[0,1]$ by \[
	d_2(x,y)=\begin{cases*}
		0&\text{if $x=y$} \\
		1&\text{otherwise}
	\end{cases*}.
\]
Note that $(2,d_2)$ is totally bounded.

Consider an instance of the situation of \cref{def:exampleKoenigMMSetting} where
\begin{itemize}
	\item the behavior functor $B$ is defined by $BX=X^\omega$,
	\item the set of modality indices $\Lambda$ is $\omega=\{0,1,2,\dots\}$, and
	\item the observation modality $\tau_i\colon [0,1]^\omega\to[0,1]$ for a modality index $i\in\omega$ is defined as the projection $\tau_i((x_j)_{j\in\omega})=x_i$.
\end{itemize}

Then the codensity lifting $\codlif{(-)^\omega}{d_e}{\tau}$ does not preserve total boundedness.
In fact, the pseudometric space $\codlif{(-)^\omega}{d_e}{\tau}(2,d_2)$ is not totally bounded.
First we show a lemma.
Let $(2^\omega,d_{2^\omega})=\codlif{(-)^\omega}{d_e}{\tau}(2,d_2)$.
\begin{lemma}
	The distance function $d_{2^\omega}$ satisfies \[
		d_{2^\omega}(x,y)=\begin{cases*}
			0&\text{if $x=y$} \\
			1&\text{otherwise}
		\end{cases*}.
	\]
\end{lemma}
\begin{proof}
	It suffices to show that, for $x\neq y$, $d_{2^\omega}(x,y)=1$.
	Let $x=(x_i)_{i\in\omega}$ and $y=(y_i)_{i\in\omega}$.
	Take $i\in\omega$ so that $x_i\neq y_i$.
	Define $f\colon 2\to[0,1]$ by $f(0)=0$ and $f(1)=1$.
	Then $f$ is a nonexpansive map from $(2,d_2)$ to $([0,1],d_e)$.
	
	Using these data, we can see
	\begin{align*}
		d_{2^\omega}(x,y)&=\sup_{g\colon(2,d_2)\to([0,1],d_e),j\in\omega}d_e(\tau_j(((g)^\omega)(x)),\tau_j(((g)^\omega)(y)))\\
		&\ge d_e(\tau_i(((f)^\omega)(x)),\tau_i(((f)^\omega)(y))) \\
		&= d_e(f(x_i),f(y_i))\\
		&=1.
	\end{align*}
\end{proof}

\begin{proposition}
	The pseudometric space $(2^\omega,d_{2^\omega})$ is not totally bounded.
\end{proposition}
\begin{proof}
	For any given $0<\varepsilon<1$, each disc of radius $\varepsilon$ covers only one point.
	This implies that finitely many such discs cannot cover the space $(2^\omega,d_{2^\omega})$, which means it is not totally bounded.
\end{proof}

This space $(2^\omega,d_{2^\omega})$ also shows that we cannot simply remove total boundedness in \cref{prop:exampleKoenigMMDensity}.
Let $\mathcal{F} = \PMetb((2^\omega,d_{2^\omega}), ([0,1],d_e))$.
Define $\mathcal{G} \subseteq \mathcal{F}$ as the set of all functions that depend only on finitely many components.
Then this $\mathcal{G}$ satisfies the two conditions in \cref{prop:exampleKoenigMMDensity}.
However, this is not dense in $\mathcal{F}$:
\begin{proposition}
	Under the topology of uniform convergence, $\mathcal{G}$ is not dense in $\mathcal{F}$.
\end{proposition}
\begin{proof}
	Define $h\colon(2^\omega,d_{2^\omega})\to([0,1],d_e)$ by: \[
		h((x_i)_{i\in\omega}) = \begin{cases*}
			1&\text{if there is infinitely many $i$'s s.t.~$x_i=1$} \\
			0&\text{otherwise}
		\end{cases*}.
	\]
	By the lemma this is indeed nonexpansive.
	
	Fix any $g\in\mathcal{G}$.
	By the definition of $\mathcal{G}$, we can take $n\in\omega$ such that $g$ only depends on the first $n$ components.
	Let $x=(0,0,\dots)\in 2^\omega$.
	Define $y\in 2^\omega$ so that the first $n$ components of $y$ are $0$ and all the others are $1$.
	Then $h(x)=0$, $h(y)=1$ and $g(x)=g(y)$ holds.
	This implies that $d_e(h(x),g(x))\ge 1/2$ or $d_e(h(y),g(y))\ge 1/2$ holds.
	In particular, the uniform distance between $h$ and $g$ is at least $1/2$.
	Since $g$ is arbitrary, $\mathcal{G}$ is not dense in $\mathcal{F}$.
\end{proof}

\end{document}